\documentclass{theoretics}
\usepackage[T1]{fontenc}
\usepackage{verbatim}
\usepackage{float}
\usepackage{amsfonts}
\usepackage{calc}
\usepackage{mathtools}
\usepackage{xspace}

\ThCSnewtheostd{hypothesis}

\newcommand{\eps}{{\varepsilon}}

\newcommand{\defproblem}[3]{

\bigskip
\noindent
\begin{center}\fbox{
	\begin{minipage}{.97\linewidth}
	\textsc{#1}
	
	\smallskip
	\noindent\begin{tabular}{@{}l@{ }l}
	\emph{Input:} & \begin{minipage}[t]{\linewidth-\widthof{Question:\ }} #2\end{minipage}\\
	\emph{Question:} & \begin{minipage}[t]{\linewidth-\widthof{Question: \ }} #3\end{minipage}
	\end{tabular}
	\end{minipage}
}
\end{center}

\medskip
}

\newcommand{\Oh}{{\cal O}}
\newcommand{\tOh}{{\widetilde {\cal O}}}
\newcommand{\algo}{{\cal A}}

\newcommand{\lemref}[1]{Lemma~\ref{lem:#1}}
\newcommand{\lemrefs}[2]{Lemmas~\ref{lem:#1} and~\ref{lem:#2}}

\newcommand{\lePara}[1][P]{\ensuremath{\le_{#1}}\xspace}
\newcommand{\bin}{\mathrm{bin}}

\newcommand{\vin}{v_\mathrm{in}}
\newcommand{\vout}{v_\mathrm{out}}
\newcommand{\uin}{u_\mathrm{in}}
\newcommand{\uout}{u_\mathrm{out}}
\newcommand{\colors}{\mathrm{col}}
\newcommand{\pathseq}{\mathrm{path}}
\newcommand{\cspecial}{\tilde{c}}
\newcommand{\clique}{\mathcal{C}}

\newcommand{\RestateRemark}[1]{{\normalfont\bfseries #1}}
\newcommand{\RestateInit}[1]{\newcommand{#1}{}}
\newcommand{\RestateGo}[1]{\renewcommand{#1}{(Restated)}}

\ThCSauthor{Karl Bringmann}{bringmann@cs.uni-saarland.de}
\ThCSaffil{Saarland University and Max-Planck-Institute for Informatics, Saarland Informatics Campus}

\ThCSauthor{Allan Gr\o nlund}{ag@kvantify.dk}
  \ThCSaffil{Aarhus University and Kvantify}

\ThCSauthor{Marvin K\"unnemann}{marvin.kuennemann@kit.edu}
\ThCSaffil{Karlsruhe Institute of Technology}

\ThCSauthor{Kasper Green Larsen}{larsen@cs.au.dk}
\ThCSaffil{Aarhus University}

\title{The NFA Acceptance Hypothesis: Non-Combinatorial and Dynamic Lower Bounds}

\ThCSthanks{This work is part of the project TIPEA (PI: Karl Bringmann) that has received funding from the European Research Council (ERC) under the European Unions Horizon 2020 research and innovation programme (grant agreement No.\ 850979). Marvin  K\"unnemann's research was partially supported by the Deutsche Forschungsgemeinschaft (DFG, German Research
  Foundation) - 462679611. Kasper Green Larsen was supported by Independent Research Fund Denmark (DFF) Sapere Aude Research Leader grant No 9064-00068B. 
This article was invited from ITCS 2024 \cite{BringmannGKL24}.}

\ThCSshortnames{K. Bringmann, A. Gr\o nlund, M. K\"unnemann, K. G. Larsen}
\ThCSshorttitle{The NFA Acceptance Hypothesis: Non-Combinatorial and Dynamic Lower Bounds}

\ThCSkeywords{fine-grained complexity theory, non-deterministic finite automata, OMv hypothesis, CFL reachability, Word Break problem}

\ThCSyear{2024}
\ThCSarticlenum{22}
\ThCSreceived{Feb 5, 2024}
\ThCSrevised{May 17, 2024}
\ThCSaccepted{May 25, 2024}
\ThCSpublished{Oct 4, 2024}
\ThCSdoicreatedtrue

\begin{document}
\maketitle

\begin{abstract}
We pose the fine-grained hardness hypothesis that the textbook algorithm for the NFA Acceptance problem is optimal up to subpolynomial factors, even for dense NFAs and fixed alphabets. 

We show that this barrier appears in many variations throughout the algorithmic literature by introducing a framework of Colored Walk problems. These yield fine-grained equivalent formulations of the NFA Acceptance problem as problems concerning detection of an $s$-$t$-walk with a prescribed color sequence in a given edge- or node-colored graph.
For NFA Acceptance on \emph{sparse} NFAs (or equivalently, Colored Walk in sparse graphs), a tight lower bound under the Strong Exponential Time Hypothesis has been rediscovered several times in recent years.
We show that our hardness hypothesis, which concerns \emph{dense} NFAs, has several interesting implications:
\begin{itemize}
	\item It gives a tight lower bound for Context-Free Language Reachability. This proves conditional optimality for the class of 2NPDA-complete problems, explaining the \emph{cubic bottleneck} of interprocedural program analysis.
	\item It gives a tight $(nm^{1/3} + m)^{1-o(1)}$ lower bound for the Word Break problem on strings of length $n$ and dictionaries of total size $m$.
	\item It implies the popular OMv hypothesis. Since the NFA acceptance problem is a static (i.e., non-dynamic) problem, this provides a static reason for the hardness of many dynamic problems.
\end{itemize}
Thus, a proof of the NFA Acceptance hypothesis would resolve several interesting barriers. Conversely, a refutation of the NFA Acceptance hypothesis may lead the way to attacking the current barriers observed for Context-Free Language Reachability, the Word Break problem and the growing list of dynamic problems proven hard under the OMv hypothesis. 
\end{abstract}


\section{Introduction}
\label{sec:intro}

Consider a classic problem that lies at the heart of introductory undergraduate courses on the theory of computation: Given a nondeterministic finite automaton (NFA) $M$ and a string~$x$ over some alphabet $\Sigma$, determine whether $x$ is accepted by $M$. The textbook algorithm for solving this problem uses dynamic programming and solves the problem in time $O(|M|\cdot |x|)$. 
Optimality of this algorithm is known for \emph{sparse} NFAs, assuming the Strong Exponential Time Hypothesis and up to subpolynomial factors, as we will discuss in detail in Section~\ref{sec:support}. 
We put forth the following fine-grained hardness hypothesis, in which we conjecture optimality of this decades-old algorithm also for \emph{dense} NFAs. 

\begin{hypothesis}[NFA Acceptance hypothesis, informal version] The textbook $O(|M|\cdot |x|)$ time algorithm for NFA Acceptance over any fixed alphabet $\Sigma$ is optimal up to subpolynomial factors, even if $M$ is a \emph{dense} NFA.
\end{hypothesis}

In this paper, we shed light on the many guises of the NFA Acceptance problem and highlight the ramifications of the above hardness hypothesis on fine-grained complexity theory in P.

\subsection{The Many Guises of the NFA Acceptance Problem}

To express equivalent formulations of the NFA Acceptance problem, consider the following problem: Given a directed, simple graph $G=(V,E)$ with edge colors $c: E\to \Sigma$, distinguished nodes $s,t\in V$ and a color sequence $c_1,\dots, c_\ell \in \Sigma$, determine whether there is a walk of length $\ell$ from $s$ to $t$ such that 
the color sequence of the traversed edges is equal to $c_1,\dots, c_\ell$. In a graph with $n$ nodes and $m\ge n$ edges, this problem is easily solved in time $O(m\ell)$, by maintaining a set $S\subseteq V$ of states reachable from $s$ via a walk with color sequence $c_1,\dots, c_i$, over all $1\le i \le \ell$.

There are natural variants of this problem, by considering \emph{node} colors rather than \emph{edge} colors, \emph{undirected} rather than \emph{directed} graphs, as well as various restrictions on the \emph{size of the alphabet}: We call these problems \textsc{Directed/Undirected $\Sigma$-Edge-Colored/$\Sigma$-Node-Colored Walk}, see Section~\ref{sec:framework} for a formal definition. All of these problems can be solved in time $O(m\ell)$. 
These Colored Walk problems arise in different communities, typically with different generalizations:

\begin{itemize}
	\item (Formal language theory:) NFA Acceptance is a fundamental problem in formal language theory, as it is the membership problem for regular languages.  \textsc{Directed $\Sigma$-Edge-Colored Walk} is precisely the special case in which the NFA has no loops (transitions from some state~$q$ to itself) or multiple transitions between any two states (with different labels from the terminal set).\footnote{Formulating the NFA Acceptance problem as a colored walk problem is not uncommon, see, e.g.,~\cite{ArenasCJR22}. Here we use the Colored Walk problem as an umbrella to express various related problems in a common language.} In Section~\ref{sec:equivalences}, we will show that this special case is fine-grained equivalent to the general case of the NFA Acceptance problem.
	\item (Hidden Markov models:) A Hidden Markov Model (HMM) is a Markov chain in which every state has a distribution over possible observations $\Sigma$ it emits. In the \textsc{Viterbi Path} problem, the task is to determine, given a sequence $c_1,\dots, c_\ell$ of observations over $\Sigma$ and an HMM $M$, the most likely walk through $M$ to emit $c_1,\dots, c_\ell$. An extension of the NFA Acceptance algorithm, well-known as Viterbi's algorithm~\cite{Viterbi67}, solves this problem in time $O(m\ell)$, where $m$ denotes the size of $M$. Exploiting the weighted nature of this problem, Viterbi's algorithm has been proven optimal up to subpolynomial factors under the weighted $k$-Clique hypothesis~\cite{BackursT17}.  \textsc{$\Sigma$-Node-Colored Walk} is equivalent to the following interesting \emph{unweighted} special case of the problem: in an HMM for which every node $v$ emits a single observation $\sigma_v$ with probability 1, determine whether there exists \emph{any} walk with positive probability for emitting the observations $c_1,\dots, c_\ell$.\footnote{The equivalence of detecting an $s$-$t$-walk with a given color sequence and detecting \emph{any} walk with a given color sequence is discussed in Section~\ref{sec:equivalences}.} 
	\item (Combinatorial Pattern matching, Bioinformatics:) In some applications in combinatorial pattern matching, e.g., in bioinformatics, it is natural to represent a \emph{set} of strings using a node-labeled graph $G$: Here, every node $u$ is equipped with a string $L_u\in \Sigma^*$ and the represented set of strings is simply the list of all concatenations $L_{u_1} \dots L_{u_k}$ over all walks $u_1 \dots u_k$ in $G$. A natural question is to perform pattern matching on these graphs, i.e., find an occurrence, or more generally, all occurrences of a given pattern $P$ in some string $T$ occurring in $G$. A line of work, see, e.g., \cite{ManberWu92, AmirLL00, Navarro00, JainZGA20} derives algorithms for exact and approximate matches of a string, see also \cite{EquiMTG23} for references. $\Sigma$-Node-Colored Walk is equivalent to detecting an exact occurrence in a graph in which every node is labeled by a single character in $\Sigma$.\footnote{See previous footnote.}
	\item (Database Theory:) An important topic in the area of graph databases is \emph{regular path queries}, see, e.g.~\cite{MendelzonW95,Barcelo13,CaselS23} and references therein. Here, a database is given as a graph $G=(V,E)$ with edge labels over $\Sigma$ and the aim is to support queries, which on input a regular expression $q$, return information about walks in the graph whose labels match the regular expression $q$ (variants of these queries include emptiness, counting, or enumeration of such walks or their endpoints). The emptiness query can be solved in time $O((|V|+|E|)|q|)$ using a straightforward approach of constructing a product graph of $G$ and an NFA representing~$q$; conditional optimality of this approach has been investigated in~\cite{CaselS23}. The special case in which $q$ describes a single string of length $\ell$ is precisely the \textsc{$\Sigma$-Edge-Colored Walk} problem.\footnote{See previous footnotes.}
\end{itemize}

Interestingly, it turns out that all the special cases discussed above are equivalent in a fine-grained sense. In particular, the NFA Acceptance problem, as well as all $n^{o(1)}$-Edge-Colored Walk and $n$-Node-Colored Walk problems are equivalent to the following seemingly simple problem:

\defproblem{Directed 2-Edge-Colored Walk}{Directed graph $G=(V,E)$, $s,t \in V$, $c:E\to \{1,2\}$, $\ell \in \mathbb{N}$, colors $c_1,\ldots,c_\ell \in \{1,2\}$}{Exist $s = v_0,v_1,\ldots,v_\ell = t$ with $(v_{i-1},v_i) \in E$ and $c(v_{i-1},v_i)= c_i$ for all $1 \le i \le \ell$?}

For a formal statement of this equivalence and its proof, see Sections~\ref{sec:framework} and~\ref{sec:equivalences}.

Thus, the same barrier has been observed across different communities, with extensions of the $O(m\ell)$-time algorithm in various distinct directions. The only generalization with a tight conditional lower bound for all graph densities is known in the weighted setting of Viterbi's algorithm~\cite{BackursT17}.

A refutation of the NFA Acceptance hypothesis would thus immediately be interesting for a host of communities -- as a case in point, several works derive essentially the same conditional lower bounds for sparse graphs and for combinatorial algorithms, see Section~\ref{sec:support}.  Furthermore, refuting the hypothesis is a prerequisite to obtaining faster algorithms also for the discussed generalizations, e.g., for \emph{approximating} the \textsc{Viterbi Path} problem.

\subsection{Support for the Hypothesis}
\label{sec:support}

The current support for the hypothesis is threefold: (I) There are tight conditional lower bounds for \emph{sparse} NFAs based on the Orthogonal Vectors Hypothesis (OVH) and thus the Strong Exponential Time Hypothesis (SETH), (II) there is a tight combinatorial lower bound for all densities of NFAs based on the combinatorial $k$-Clique hypothesis, and (III) simply the lack of a faster algorithms despite the centrality and age of this problem. We discuss these reasons in detail:

\paragraph{Support I: Hardness for Sparse NFAs}

It is known that an $O( (n\ell)^{1-\epsilon})$-time algorithm for acceptance of a length-$\ell$ string by an $n$-state NFA would refute OVH and SETH. This can be viewed as a tight conditional lower bound for \emph{sparse} NFAs with $m=O(n)$ transitions, but it fails to provide a tight lower bound for NFAs with $m=\Theta(n^\alpha)$ transitions with $1 < \alpha \le 2$. Several versions of this result have been proven and published in recent years. A subset of the authors have learned the result phrased for NFAs from Russell Impagliazzo in 2015. Since regular expressions of size $n$ can be easily converted to an NFA with $O(n)$ states and transitions, the quadratic conditional lower bounds for regular expression membership/matching in~\cite{BackursI16,BringmannGL17} translate to NFA Acceptance lower bounds. Subsequently, the lower bound for sparse NFAs has been given in the language of string matching in labeled graphs in~\cite{EquiMTG23}, in the language of NFA Acceptance in~\cite{PotechinS20}, and in the language of regular path queries in graph databases in~\cite{CaselS23}.\footnote{By the  the fine-grained equivalence of the Colored Walk problems of Section~\ref{sec:framework}, these results turn out to be essentially equivalent in their main statement. However, some aspects of the reductions can differ.}

Interestingly, the same lower bound already holds under a weaker hypothesis, specifically Formula-SETH~\cite{AbboudB18, Schepper20, GibneyHT21}.

\paragraph{Support II: Combinatorial Hardness}

Beyond sparse NFAs, the literature provides, perhaps a bit implicitly, a (combinatorial) reduction from $k$-Clique to NFA Acceptance with a dense NFA with $n$ states with the following implications: (1) An $O( (\ell n^2)^{1-\epsilon} )$-time \emph{combinatorial}\footnote{Here, a \emph{combinatorial} algorithm refers to an algorithm that avoids the algebraic techniques underlying fast matrix multiplication algorithms, see Section~\ref{sec:non-combinatorial} for a discussion.} NFA Acceptance algorithm for any $\ell=n^\beta, \beta > 0$ would give a \emph{combinatorial} $O(n^{(1-\epsilon')k})$-time $k$-Clique algorithm for all sufficiently large $k$, and (2) an $O( (\ell n^{2})^{\omega/3-\epsilon} )$-time  NFA Acceptance algorithm would give an $O(n^{(\omega/3)k-\epsilon'})$-time $k$-Clique algorithm for sufficiently large $k$, breaking the state of the art for current $k$-Clique algorithms, see~\cite{NesetrilP85,AbboudBVW18}.

These lower bounds can be obtained either by adapting the reduction from weighted $k$-Clique in~\cite{BackursT17} to the unweighted case, or as a special case of~\cite[Theorem I.5]{AbboudBBK17}\footnote{In the theorem, simply choose $\alpha_N=\alpha_n=\beta$.}, or as an appropriate generalization of the reduction from Triangle Detection to NFA Acceptance with $\ell \approx n$ in~\cite{PotechinS20}, or as an appropriate generalization of the reduction from Triangle Detection to regular path queries with queries of length $\ell\approx n$ in~\cite{CaselS23}. For a self-contained reduction see Section~\ref{sec:otherHypos}.

These reductions suggest that to break the NFA Acceptance hypothesis, a non-trivial application of fast matrix multiplication techniques is needed. Put differently, our hypothesis boils down to postulating that \emph{fast matrix multiplication is not applicable for the NFA Acceptance problem}. 

\paragraph{Support III: Long-standing State of the Art}

Finally, the NFA Acceptance hypothesis is plausible simply due to the lack of any improved algorithms. To the best of our knowledge, the best algorithm for NFA Acceptance for dense NFAs over a constant-sized alphabet $\Sigma$ runs in time $n^2\ell/2^{\Omega(\sqrt{\log n})}$ (see Section~\ref{sec:omv}).

Let us contrast the state of the art for NFA Acceptance, i.e., the membership problem for regular languages, to the membership problem for context-free languages, specifically context-free grammar parsing. The fine-grained complexity of this problem has parallels to NFA Acceptance: The best combinatorial algorithms decide membership of a length-$n$ string in a fixed grammar $\Gamma$ in mildly subcubic time $O(n^3/\text{polylog } n)$, see \cite{AbboudBVW18} for an overview. Any truly subcubic combinatorial algorithm would refute the combinatorial $k$-Clique algorithm~\cite{AbboudBVW18}. However, by a highly non-trivial application of fast matrix multiplication, Valiant~\cite{Valiant75} gave an algorithm for context-free grammar parsing in time $O(n^\omega)$ for constant-sized grammars, but this is known already since 1975! It appears unlikely that a similar algorithm using fast matrix multiplication for NFA Acceptance has been overlooked for decades.

In Section~\ref{sec:natural-approaches}, we also discuss why natural algorithmic approaches towards refuting the NFA Acceptance hypothesis appear to fail.

\medskip
It is apparent that the NFA Acceptance hypothesis poses a significant barrier.
Finally, we remark that \cite{PotechinS20} recently posed our NFA Acceptance hypothesis as an open problem.\footnote{Potechin and Shallit note, following the SETH-based lower bound for sparse NFAs: ``However, this does not rule out
	an improvement when the NFA is dense, and we leave it as an open problem to either find a significant improvement to this algorithm, or show why such an improvement is unlikely.''}

\subsection{Evidence Against the Hypothesis}

The only evidence against the NFA Acceptance hypothesis that we are aware of is the existence of (co-)nondeterministic verifiers with running time $O( (n^2 \ell)^{1-\epsilon})$. This is even known for the weighted generalization to the Viterbi Path problem~\cite{BackursI16}, but can be proven directly (see Section~\ref{sec:otherHypos}) or follows as a consequence of our reduction in Section~\ref{sec:cfl-reach} combined with~\cite{ChistikovMS22}. This rules out tight \emph{deterministic} reductions from SAT and OV to NFA Acceptance with dense NFAs, assuming a nondeterministic variant of SETH~\cite{CarmosinoGIMPS16}. However, there are several hypotheses with a similar status, in particular, the APSP hypothesis, 3SUM hypothesis and the Hitting Set hypothesis.

\subsection{Applications I:  Non-combinatorial Lower Bounds}
\label{sec:non-combinatorial}

As our first major consequence of the NFA Acceptance hypothesis, we give novel tight conditional lower bounds against general algorithms where previously only combinatorial lower bounds were known.

\paragraph{Combinatorial Lower Bounds}
Some of the oldest conditional lower bounds in the polynomial-time regime are based on Boolean matrix multiplication (BMM), see, e.g.~\cite{Lee02} (for a survey of fine-grained complexity theory in P, we refer to~\cite{VassilevskaW18}). The best known algorithm for BMM uses fast matrix multiplication over the integers, yielding an $n^{\omega+o(1)}$-time algorithm, where $\omega$ is the matrix multiplication exponent, with a current state-of-the-art bound of $\omega < 2.371552$ due to works of~\cite{Strassen69, CoppersmithW90, VassilevskaW12, LeGall14, AlmanVW24, DuanWZ23, VWilliamsXXZ24}, among others.

Unfortunately, theoretically fast matrix multiplication algorithms generally turn out to be impractical, as the constants hidden in the $O$-notation are usually too large for the problem sizes observed in applications. Therefore, it has become increasingly popular to use the notion of \emph{combinatorial} algorithms to exclude the use of algebraic techniques underlying current fast matrix multiplications, in an attempt to study algorithms whose asymptotic complexity translates into practical running time. In this direction, it has been hypothesized that BMM does not admit combinatorial $O(n^{3-\epsilon})$-time algorithms with $\epsilon> 0$. This hypothesis is in fact equivalent to the nonexistence of truly subcubic combinatorial algorithms for triangle detection, and gives tight combinatorial lower bounds for problems such as the Sliding Window Hamming Distance (see~\cite{GawrychowskiU18}), CFL Reachability~\cite{ChatterjeeCP18} and many more. 
A generalization of this hypothesis is that for no $k\ge 3$ and $\epsilon > 0$, $k$-Clique admits a combinatorial $O(n^{k-\epsilon})$-time algorithm, which implies a combinatorial version of our NFA Acceptance hypothesis, see Section~\ref{sec:otherHypos}.

However, there are several downsides to the notion of \emph{combinatorial} lower bounds: First and foremost, the notion of \emph{combinatorial algorithm} is not formally defined, which is hardly acceptable (see e.g.\ the discussion in~\cite{HenzingerKNS15}). 
Second, in theoretical computer science we are usually interested in the optimal asymptotic worst-case complexity, so why would we exclude algorithmic techniques leading to an improved complexity?
Third, Strassen's original algorithm is sometimes found to be practical~\cite{HuangSHG16}. 
Hence, there are several reasons to avoid ``combinatorial lower bounds'', as they leave open the very real possibility that faster algorithms exploiting fast matrix multiplication techniques exist. 

How can we distinguish between problems for which faster algorithms via fast matrix multiplications exist, and those for which such an improvement is unlikely?
Generally speaking, the current state of the art for problems with tight combinatorial lower bounds can be categorized as follows:
\begin{enumerate}
\item (easy:) Problems whose asymptotic complexity is upper bounded by the complexity of fast matrix multiplication, e.g., context-free grammar recognition for constant-sized grammars~\cite{Valiant75}, maximum node-weighted triangle~\cite{CzumajL09}, and others.
\item (intermediate:) Problems for which improvements via fast matrix multiplication are possible, but they do not necessarily lead to quasilinear time in the input if $\omega=2$: e.g., Sparse Triangle Listing or intermediate $O(n^{\frac{3+\omega}{2}})$-time problems such as $(\min,\max)$ product, Dominance product, Equality product, and All-Edges Monochromatic Triangle. Only for a few of these problems, conditional lower bounds give evidence why they do not appear to belong to the first category, e.g.~\cite{Patrascu10, LincolnPVW20, VassilevskaWX20, ChanVWX21}.
\item (difficult:) Problems for which even using fast matrix multiplication, no improved algorithms are known, e.g., APSP, Sliding Window Hamming Distance, Word Break, Context-free Language Reachability, Klee's Measure Problem. To establish that a problem lies in this category, currently only two hardness assumptions appear applicable: If the problem can express \emph{weights}, a reduction from APSP or more generally the weighted $k$-Clique hypothesis is plausible, e.g.~\cite{VassilevskaWW18,BackursDT16,BackursT17}. A second possibility is a reduction based on the 3-uniform Hyperclique hypothesis~\cite{LincolnVWW18}, which was used, e.g., for Klee's measure problem in 3D~\cite{Kunnemann22}.
\end{enumerate}

The NFA Acceptance hypothesis yields a new route to establish that a problem belongs to the third category. We describe two such applications: (1) Context-free language reachability and, more generally, all 2NPDA-complete problems (which suffer from the \emph{cubic bottleneck} in interprocedural program analysis), and (2) the Word Break problem. 

\paragraph{CFL Reachability and 2NPDA-hard problems}
In Context-free Language (CFL) Reachability for a fixed context-free grammar $\Gamma$, we are given a $\Sigma$-edge-colored directed graph $G$ and distinguished nodes $s,t$, and the task is to determine whether there is a walk from $s$ to $t$ such that the sequence of traversed edge colors spells a word in the language given by $\Gamma$. The CFL Reachability problem has many applications in program analysis, verification and database theory (see, e.g., the references in~\cite{ChistikovMS22,KoutrisD23}). A classic algorithm solves CFL Reachability in time $O(n^3)$~\cite{Yannakakis90}, which has been slightly improved to time $O(n^3 / \log n)$~\cite{Rytter85,Chaudhuri08}. In fact, these algorithms extend to the generalization of computing all pairs $(s,t)$ for which such a walk exists.

The inability to obtain truly subcubic time algorithms for CFL Reachability and related problems has led researchers to investigate its relationship to the recognition problem for two-way nondeterministic pushdown automata (2NPDA), which admits a classic cubic time algorithm due to~\cite{AhoHU68}. Several problems, including CFL Reachability, pushdown automata emptiness, and related problems in data flow analysis, have been proven 2NPDA-complete in the sense that they are subcubic equivalent to 2NPDA recognition~\cite{Neal89,MelskiR00,HeintzeM97}, see also~\cite{ChistikovMS22}.

Since then, there have been attempts to obtain conditional lower bounds for 2NPDA-complete problems: For \emph{combinatorial} algorithms a conditional lower bound ruling out time $O(n^{3-\eps})$ for any $\eps > 0$ was shown in~\cite{ChatterjeeCP18}. However, it was proven that obtaining a tight lower bound based on SETH/OVH using deterministic reductions would violate NSETH~\cite{ChistikovMS22}. Recently, an interesting attempt was made to show that the All-Pairs version of CFL Reachability is strictly harder than matrix multiplication (i.e., does not belong to the first category in the above list) via a reduction from All-Edges Monochromatic Triangle~\cite{KoutrisD23}.\footnote{Unfortunately, the reduction from Monochromatic Triangle to All-Pairs CFL Reachability given in~\cite{KoutrisD23} appears to be flawed. In this reduction, a graph is created where every edge in the original graph is replaced by a line graph, which in general leads to a quadratic blow-up in the number of nodes. In personal communication, the authors confirm this issue.}

In Section~\ref{sec:cfl-reach}, we show that proving the NFA Acceptance hypothesis settles the cubic bottleneck of 2NPDA-complete problems by showing that the NFA Acceptance hypothesis implies that CFL Reachability, and thus, all 2NPDA-hard problems, have no truly subcubic algorithms. In particular, this gives conditional optimality of the algorithms due to Yannakakis~\cite{Yannakakis90} and Aho et al.~\cite{AhoHU68}.   

\paragraph{Word Break}
In the Word Break problem we are given a string $S$ and a dictionary $D$ (i.e., $D$ is a set of strings) and we ask whether $S$ can be split into dictionary words, i.e., whether we can partition $S$ into substrings such that each substring is in the set $D$. Denoting the length of $S$ by $n$ and the total length of all dictionary words by $m$, this problem has a relatively simple randomized algorithm running in time $\tOh(n m^{1/2} + m)$.\footnote{Throughout the paper we write $\tOh(T) := \bigcup_{c \ge 0} O(T \log^c T)$, where we denote $\log^c T = (\log(T))^c$.} This running time was first improved to $\tOh(n m^{1/2 - 1/18} + m)$~\cite{BackursI16} and then to the current state of the art $O(n (m \log m)^{1/3} + m)$~\cite{BringmannGL17}. In~\cite{BringmannGL17} it was also shown that Word Break has no combinatorial algorithm running in time $O(n m^{1/3 - \eps} + m)$ for any $\eps>0$, assuming a hypothesis on combinatorial $k$-Clique. Under standard, non-combinatorial hypotheses, no superlinear running time lower bound is known.

In Section~\ref{sec:wordbreak}, we settle this issue by proving a tight conditional lower bound for Word Break under the NFA Acceptance hypothesis.

\subsection{Application II: A Static Reason for Dynamic Hardness}

In 2015, Henzinger, Forster, Nanongkai and Saranurak~\cite{HenzingerKNS15} formulated the OMv hypothesis which unifies and strengthens many conditional lower bounds for data structure problems and dynamic problems and has since been used in various areas such as graph algorithms~\cite{HenzingerKNS15,Dahlgaard16,AbboudD16,HenzingerW21}, string algorithms~\cite{CliffordGLS18,KempaK22,CliffordGK0U22}, computational geometry~\cite{DallantI22,LauR21}, linear algebra~\cite{BrandNS19,JiangPW23}, formal languages~\cite{GibneyT21}, and database theory~\cite{BerkholzKS17,CaselS23,KaraNOZ23}.
In fact, almost all conditional lower bounds for dynamic problems known under BMM or other hypotheses such as 3SUM or APSP can also be shown under the OMv hypothesis.\footnote{Lower bounds from SETH seem to be an exception from this rule of thumb.}
Thus, there is one reason that allows to rule out faster algorithms for a wealth of problems, without restricting the class of algorithms that may be used, and without any undefined notions.

Interestingly, it turns out that \emph{the NFA Acceptance hypothesis implies the OMv hypothesis}. The reduction is very direct and has nice implications: (1) It establishes NFA Acceptance as a static reason for the hardness of a wealth of dynamic problems. (2) It gives additional reason to believe in the OMv hypothesis, since NFA Acceptance is an additional problem that has been studied for decades and lacks polynomially improved algorithms.
(3) We obtain an improved running time for NFA Acceptance by a factor $2^{\Theta(\sqrt{\log n})}$, by the improved algorithm for OMv~\cite{LarsenW17}.
(4) Since NFA Acceptance is a static problem, we can try to obtain similar lower bounds as from OMv now for static versions of dynamic problems.
We give the reduction in Section~\ref{sec:omv}.

\section{Equivalent formulations: Colored Walk Framework}
\label{sec:framework}

An NFA $M$ is a tuple $(Q, \Sigma, \delta, q_0, F)$ with $\delta \subseteq Q \times \Sigma \times Q$, $q_0 \in Q$ and $F \subseteq Q$. We say that $M$ accepts a string $x\in \Sigma^*$ if and only if there exists some sequence $q_1,\dots, q_{|x|}\in Q$ with $(q_i, x_{i+1}, q_{i+1})\in \delta$ for all $0\le  i < |x|$ and $q_{|x|}\in F$. Note that we do not allow $\epsilon$-transitions.\footnote{An $\epsilon$-transition is a transition labeled with the empty word $\epsilon$. Such a transition can be taken by the NFA without reading a character of $x$. While for dense NFAs, allowing $\epsilon$-transitions yields a linear-time equivalent version of the NFA acceptance problem, we leave it as an open problem whether $\epsilon$-transitions make the problem more difficult for sparser NFAs.} The NFA Acceptance problem asks to determine, given an NFA $M$ and string $x$ over $\Sigma$, whether $M$ accepts $x$.
We usually use $n=|Q|$ to denote the number of states, $m=|\delta|$ to denote the number of transitions and $\ell=|x|$ to denote the length of the string $x$.

We are ready to state our hypothesis formally:

\begin{hypothesis}[NFA Acceptance hypothesis, formal version] Let $1\le \alpha \le 2$ and $\beta, \epsilon > 0$ be arbitrary. There is no (randomized) algorithm that solves the NFA Acceptance problem for NFAs with $n$ states and $m=\Theta(n^\alpha)$ transitions, alphabet $\Sigma$ of size $n^{o(1)}$ and strings of length $\ell = \Theta(n^{\beta})$ in time $O((m\ell)^{1-\epsilon}) = O(n^{(\alpha+\beta)(1-\epsilon)})$.
\end{hypothesis}

It is equivalent to state the hypothesis only for dense NFAs, i.e., $\alpha = 2$, as we show next.

\begin{lemma}\label{lem:dense-is-hardest}
The NFA Acceptance hypothesis is equivalent to the following statement: Let $\beta, \epsilon > 0$ be arbitrary. There is no (randomized) algorithm that solves the NFA Acceptance problem for NFAs with $n$ states and $m=\Theta(n^2)$ transitions, alphabet $\Sigma$ of size $n^{o(1)}$ and strings of length $\ell = \Theta(n^{\beta})$ in time $O((n^2 \ell)^{1-\epsilon}) = O(n^{(2+\beta)(1-\epsilon)})$.
\end{lemma}
\begin{proof}
Assume that there exists $1\le \alpha' < 2$ and $\beta'> 0$ such that we can solve the problem with $m=\Theta(n^{\alpha'})$ transitions and sequence length $\ell = \Theta(n^{\beta'})$ in time $O(n^{(\alpha'+\beta')(1-\epsilon)})$. We define $\alpha=2$, $\beta=(2/\alpha')\beta'$ and consider any instance with $m=\Theta(n^\alpha)=\Theta(n^2)$ and $\ell = \Theta(n^\beta)$. Create an equivalent instance by simply introducing $n^{2/\alpha'} \ge n$ additional isolated states. This instance has $n'=\Theta(n^{2/\alpha'})$ states, the same number of transitions $m'=m=\Theta(n^2)=\Theta((n')^{\alpha'})$ and the same sequence length $\ell'=\ell = \Theta(n^\beta)=\Theta((n')^{\beta'})$. By our assumption, we can solve this instance in time $O((n')^{(\alpha'+\beta')(1-\epsilon)})=O(n^{(\alpha+\beta)(1-\epsilon)})$. This contradicts the assumption for $\alpha=2$ and an appropriate $\beta$, as desired.
\end{proof}

Furthermore, our conditional lower bounds for CFL Reachability and OMv follow already from the setting $\alpha = 2$ and $\beta=1$, which can be viewed as the core setting of the NFA Acceptance hypothesis: Determining whether a given $n$-state NFA $M$ over an $n^{o(1)}$-sized alphabet $\Sigma$ accepts a given length-$n$ string $x$ requires time $n^{3-o(1)}$ in the worst case.

In the remainder of the section, we discuss several equivalent formulations of the NFA Acceptance hypothesis, using a framework of graph problems that we refer to as Colored Walk problems.

\subsection{Colored Walk Framework}

\label{sec:CWvariants}

We study the following variants of the Colored Walk problem.

\defproblem{Directed/Undirected Node-$C(n)$-Colored Walk}{directed/undirected graph $G=(V,E)$ with $n = |V|$ and $m = |E|$, vertices $s,t \in V$, \\ coloring $c \colon V \to \{1,\ldots,C(n)\}$, integer $\ell$, color sequence $c_1,\ldots,c_\ell$}{Exist $s = v_0,v_1,\ldots,v_\ell = t$ with $(v_{i-1},v_i) \in E$ and $c(v_i) = c_i$ for all $1 \le i \le \ell$?}

\vspace{-0.5cm}

\defproblem{Directed/Undirected Edge-$C(n)$-Colored Walk}{directed/undirected graph $G=(V,E)$ with $n = |V|$ and $m = |E|$, vertices $s,t \in V$, \\ coloring $c \colon E \to \{1,\ldots,C(n)\}$, integer $\ell$, color sequence $c_1,\ldots,c_\ell$}{Exist $s = v_0,v_1,\ldots,v_\ell = t$ with $(v_{i-1},v_i) \in E$ and $c(v_{i-1},v_i) = c_i$ for all $i$?}

Note that all of these problems can be solved in time $O(m\ell)$.
We show that most of these problem variants are equivalent, in the following sense.
Note that statement A1 below is exactly the NFA Acceptance hypothesis.

\RestateInit{\restateEquivalences}
\begin{restatable}[Colored Walk Framework]{lemma}{equivalences}
  \label{lem:equivalencehypo}\RestateRemark{\restateEquivalences}
	Let $\alpha \in [1,2], \beta > 0$. All of the following statements A1, A2, A3, A4, A5, B1, B2, B3, and B4 are equivalent:
	\begin{enumerate}
		\item[A] Restricted to instances with $m = \Theta(n^\alpha)$ and $\ell = \Theta(n^\beta)$, there is no $O(n^{\alpha+\beta-\eps})$-time algorithm for any $\eps>0$ for the problem...
		\begin{enumerate}[1]
			\item ... NFA Acceptance with alphabet size $n^{o(1)}$.
			\item ... Directed Node-2-Colored Walk.
			\item ... Directed Node-$n$-Colored Walk.
			\item ... Directed Edge-2-Colored Walk.
			\item ... Directed Edge-$n^{o(1)}$-Colored Walk.
		\end{enumerate}
		\item[B] Restricted to instances with $m = \Theta(n^\alpha)$ and $\ell = O(n^\beta)$, there is no $O(n^{\alpha+\beta-\eps})$-time algorithm for any $\eps>0$ for the problem...
		\begin{enumerate}[1]
			\item ... Undirected Node-2-Colored Walk.
			\item ... Undirected Node-$n$-Colored Walk.
			\item ... Undirected Edge-2-Colored Walk.
			\item ... Undirected Edge-$n^{o(1)}$-Colored Walk.
		\end{enumerate}
	\end{enumerate}
\end{restatable}

The proof is deferred to Section~\ref{sec:equivalences}.

From now on, we use the term \emph{Colored Walk} for any of the above problems. In particular, in the reductions based on the NFA Acceptance hypothesis, we can always use the variant that is easiest to work with. In particular, all of our reductions in the following sections will start from Directed Edge-2-Colored Walk (i.e., we will use statement A4). By the above lemma, we thus obtain conditional lower bounds under the NFA Acceptance hypothesis.

We leave it as an open problem whether Directed/Undirected Edge-$C(n)$-Colored Walk for~$C(n) \gg n^{o(1)}$ is also equivalent to the above problems.

In Section~\ref{sec:equivalences}, we also show equivalence of Colored Walk to a version without source and target nodes $s,t$, specifically, the version in which we determine existence of \emph{any} walk in the graph with a given color sequence $c_1,\dots, c_\ell$. Such versions occur e.g., for string matching in labeled graphs or regular path queries.

\section{Hardness of CFL Reachability}
\label{sec:cfl-reach}

In this section, we give a tight hardness result for CFL reachability under the NFA acceptance hypothesis.

A \emph{context-free grammar} $\Gamma$ is a tuple $(N, \Sigma, P, S)$, where $N$ is a set of nonterminals, $\Sigma$ is a set of terminals disjoint from $N$, $S\in N$ is the start symbol and $P$ is a set of production rules of the form $X \to \alpha$ with $X\in N$ and $\alpha \in (N\cup \Sigma)^*$. For a string $x\in (N\cup \Sigma)^*$, replacing an occurrence of a nonterminal $X\in N$ by $\alpha$ (where $X\to \alpha$ is in $P$) is called an application of the production rule $X\to \alpha$. We say that a string $x \in \Sigma^*$ is generated by $\Gamma$ if $x$ can be obtained from $S$ by (repeated) applications of production rules. We let $L(\Gamma)$ denote the set of strings generated by~$\Gamma$.  

Recall that in the CFL Reachability problem we are given a context-free grammar $\Gamma$ describing a language $L(\Gamma)$ over a terminal set $\Sigma$ as well as a directed graph $G = (V,E)$ with designated vertices $s,t \in V$ where every edge $e \in E$ is labeled by a terminal $\sigma(e) \in \Sigma$. We call any sequence of terminals $w \in \Sigma^*$ a \emph{word}, and we say that a walk $v_0,v_1,\ldots,v_t$ in $G$ \emph{spells the word} $\sigma(v_0,v_1) \sigma(v_1,v_2) \ldots \sigma(v_{\ell-1},v_\ell)$, i.e., we concatenate all edge labels along the walk.
The task is to decide whether there is a walk from $s$ to $t$ in $G$ spelling a word in $L(\Gamma)$, i.e., whether there is a walk $s = v_0,v_1\ldots,v_\ell = t$ such that $\sigma(v_0,v_1) \sigma(v_1,v_2) \ldots \sigma(v_{\ell-1},v_\ell) \in L(\Gamma)$.
We write $n = |V|, m = |E|$ and we assume that $\Gamma$ is fixed, in particular it has constant size and any running time dependence on the size of $\Gamma$ can be ignored. 

The CFL Reachability problem has a classic algorithm running in time $O(n^3)$~\cite{Yannakakis90}, which has been slightly improved to time $O(n^3 / \log n)$~\cite{Chaudhuri08}.
For \emph{combinatorial} algorithms a conditional lower bound ruling out time $O(n^{3-\eps})$ for any $\eps > 0$ was shown in~\cite{ChatterjeeCP18}.

We prove a tight lower bound for CFL Reachability under the NFA Acceptance hypothesis:

\begin{theorem} \label{thm:cflreach}
There is a fixed grammar $\Gamma$ such that the CFL Reachability problem on $\Gamma$ has no $O(n^{3-\eps})$-time algorithm for any $\eps>0$ assuming the NFA Acceptance hypothesis.
\end{theorem}

\begin{proof}
We reduce from Directed Edge-2-Colored Walk. To this end, we are given a directed graph $G=(V, E)$ with colors $c \colon E \to \{1,2\}$, designated vertices $s,t\in V$, and a color sequence $c_1, \dots, c_\ell \in \{1,2\}$. Let $n = |V|$ and $m = |E|$.
We assume the graph to be dense (i.e., $m = \Theta(n^2)$) and we assume $\ell = \Theta(n)$. 
By the NFA Acceptance hypothesis in the setting $\alpha=2, \beta=1$, there is no algorithm solving all such instances of Directed Edge-2-Colored Walk in time $O(n^{3-\eps})$ for any~$\eps > 0$.

We construct an instance of CFL Reachability as follows. The context-free grammar $\Gamma$ is the language of well-formed expressions on two types of parenthesis (also known as Dyck-2) given by the nonterminal $S$, the terminals $(_1,)_1,(_2,)_2$, and the production rules $S \to SS, S \to (_1S)_1, S \to (_1)_1, S \to (_2S)_2, S \to (_2)_2$. Note that the grammar is independent of the input size, i.e., $\Gamma$ has constant size.\footnote{This special case of the CFL Reachability problem, where $\Gamma$ is Dyck-2, is also known as Dyck-2 Reachability.}

We construct a directed graph $G' = (V',E')$ by starting from the graph $G$, adding a directed path of length $\ell$ on nodes $u_0,u_1,\ldots,u_\ell$, and identifying the nodes $t$ and $u_0$. In other words, we attach a path $u_1,\ldots,u_\ell$ to the node $t$, and we use $u_0$ as another name for node $t$. We set $s' := s$ and $t' := u_\ell$. For each edge $(u,v) \in E$ we set the label to $\sigma(e) := (_{c(u,v)}$. For each $1 \le i \le \ell$ we set the label $\sigma(u_{i-1},u_i) := \; )_{c_{\ell+1-i}}$. This finishes the construction of the CFL Reachability instance $(\Gamma,G',s',t',\sigma)$.

Note that any walk from $s'$ to $t'$ in $G'$ ends with the path $u_0,u_1,\ldots,u_\ell$ and thus with the labels $)_{c_\ell} )_{c_{\ell-1}} \ldots )_{c_2} )_{c_1}$. Since there are no other edge labels with closing brackets, we must choose a walk from $s'=s$ to $u_0=t$ spelling the word $(_{c_1} (_{c_2} \ldots (_{c_{\ell-1}} (_{c_\ell}$ in order to match all parentheses. Such a walk corresponds to a walk from $s$ to $t$ in $G$ with color sequence $c_1,c_2,\ldots,c_\ell$. This shows that the constructed CFL Reachability instance is a YES-instance if and only if the given Colored Walk instance is a YES-instance, and thus shows correctness of the reduction.

Note that the constructed graph $G'$ consists of $n+\ell = O(n)$ nodes (using our assumption on~$\ell$). Therefore, any algorithm for CFL Reachability running in time $O(n^{3-\eps})$ solves the given Colored Walk instance in time $O(n^{3-\eps})$. This contradicts the NFA Acceptance hypothesis (as discussed in the first paragraph).
\end{proof}

\section{Hardness of Word Break}
\label{sec:wordbreak}

Recall that in the Word Break problem, we are given a string $S$ and a dictionary $D$ (i.e., $D$ is a set of strings) and we ask whether $S$ can be split into dictionary words, i.e., whether we can partition $S$ into substrings such that each substring is in the set $D$. We denote the length of $S$ by $n$ and the total length of all strings in $D$ by $m$. In this section, we prove that the NFA Acceptance hypothesis implies optimality (up to subpolynomial factors) of the $O(n (m \log m)^{1/3} + m)$ time algorithm given in~\cite{BringmannGL17}.

\begin{theorem} \label{thm:wordbreak}
The Word Break problem has no $O(n m^{1/3 - \eps} + m)$-time algorithm for any $\eps > 0$ assuming the NFA Acceptance hypothesis. This even holds restricted to $m = \Theta(n^{\gamma})$ for any constant $\gamma \in (0,3/2)$.
\end{theorem}
Note that the $O(n (m \log m)^{1/3} + m)$ algorithm from~\cite{BringmannGL17} solves Word Break in time $\tOh(m)$ whenever $m = \Omega(n^{3/2})$. Since this is near-linear, there is no need for proving a fine-grained lower bound in this case. Therefore, it is natural that we assume $\gamma < 3/2$ in the above theorem.

\begin{proof}[Proof of Theorem~\ref{thm:wordbreak}]
In the following proof, we use the upper case letters $N$ and $M$ for the parameters of Word Break in order to differentiate them from the parameters of Colored Walk.

We reduce from Directed Edge-2-Colored Walk. To this end, we are given a directed graph $G=(V, E)$ with colors $c \colon E \to \{1,2\}$, designated vertices $s,t\in V$, and a color sequence $c_1, \dots, c_\ell \in \{1,2\}$. We write $V = \{1,\dots, n\}$ and let $m = |E|$.
We assume the graph to be dense (i.e., $m = \Theta(n^2)$) and we assume $\ell = \Theta(n^\beta)$ for $\beta := 3/\gamma - 1$. 
By the NFA Acceptance hypothesis in the setting $\alpha=2, \beta = 3/\gamma - 1$, there is no algorithm solving all such instances of Directed Edge-2-Colored Walk in time $O(n^{2+\beta-\eps})$ for any $\eps > 0$.


We use the notation $0^k$ to denote the string of length $k$ consisting of $k$ times the letter 0.
We construct the string $S$ and the dictionary $D$ over alphabet $\{0,1,2\}$ as follows:
\begin{align*}
S &:= 0^{s} \, c_1 \, 0^{n} \, c_2 \, 0^{n} \, c_3 \, \ldots \, 0^{n} \, c_{\ell-1} \, 0^{n} \, c_\ell \, 0^{n-t}. \\
D &:= \{ \, 0^{u} \, c(u,v) \, 0^{n-v} \, \mid \, (u,v) \in E \}. 
\end{align*}

\paragraph{Correctness} We claim that the string $S$ can be split into dictionary words if and only if there is a walk from $s$ to $t$ with color sequence $c_1,\ldots,c_\ell$ in $G$. This is straightforward to show: To match the prefix $0^s c_1$ we must choose a dictionary word corresponding to an edge $(s,v_1)$ of color $c_1$, and the remaining string starts with the prefix $0^{v_1} c_2$. Generally, in the $i$th step the remaining string is of the form $0^{v_i} c_i 0^n c_{i+1} \ldots 0^n c_\ell 0^{n-t}$, so in order to match the prefix $0^{v_i} c_i$ we must choose a dictionary word corresponding to an edge $(v_i,v_{i+1})$ of color $c_i$. The last vertex $v_\ell$ must satisfy $v_\ell = t$ in order to match the suffix $0^{n-t}$. Hence, any valid partitioning of the string $S$ into dictionary words corresponds to a walk $s = v_0,v_1,\ldots,v_\ell = t$ with color sequence $c_1,\ldots,c_\ell$ in $G$, and this is an equivalence.

\paragraph{Running Time} The length of the string $S$ is $N = \Theta(\ell n) = \Theta(n^{1+\beta})$, by the assumption on $\ell$.
The total length of all strings in $D$ is $M = \Theta(nm) = \Theta(n^3)$, by the assumption that $G$ is dense. Note that $M = \Theta(N^{3/(1+\beta)}) = \Theta(N^{\gamma})$ by our setting of $\beta = 3/\gamma - 1$, so we constructed instances with the desired setting of parameters.
If Word Break can be solved in time $O(N M^{1/3 - \eps} + M)$, then by plugging in the bounds on $N$ and $M$, our setting of Directed Edge-2-Colored Walk can be solved in time $O(n^{1+\beta} (n^3)^{1/3-\eps} + n^3) = O(n^{2+\beta-3\eps} + n^3)$. Since $\gamma < 3/2$ we have $\beta = 3/\gamma - 1 > 1$ and thus by setting $\eps' := \min\{3\eps, 2+\beta-3\} > 0$ we can bound the running time by $O(n^{2+\beta-\eps'})$. This contradicts the NFA Acceptance hypothesis in the setting $\alpha=2, \beta = 3/\gamma - 1$.
\end{proof}

\section{Hardness of OMv}
\label{sec:omv}

In the Online Boolean Matrix-Vector Multiplication (OMv) problem, an algorithm is initially given an $n \times n$ Boolean matrix $M$. Then the following repeats for $n$ rounds: In the $i$th round, the algorithm is given an $n$-dimensional Boolean vector $v_i$ and has to compute $M v_i$. The algorithm must compute the output $M v_i$ before it can proceed to the next round. The running time of an OMv algorithm is the total running time for the initialization and all $n$ rounds together. 

By naively multiplying $M v_i$ in time $O(n^2)$ in each of the $n$ rounds, OMv can be solved in time $O(n^3)$. This running time has been improved to $n^3 / 2^{\Omega(\sqrt{\log n})}$~\cite{LarsenW17}. The OMv Hypothesis, due to Henzinger et al.~\cite{HenzingerKNS15}, postulates that OMv cannot be solved in strongly subcubic time.

\begin{hypothesis}[OMv Hypothesis] \label{conj:omv}
For any constant $\eps > 0$, OMv has no $O(n^{3-\eps})$-time algorithm (with an error probability of at most 1/3).
\end{hypothesis}

Many data structure problems and dynamic problems have matching lower bounds under the OMv Hypothesis, as its usefulness has been established in various areas such as graph algorithms \cite{HenzingerKNS15,Dahlgaard16,AbboudD16,HenzingerW21}, string algorithms~\cite{CliffordGLS18,KempaK22,CliffordGK0U22}, computational geometry~\cite{DallantI22,LauR21}, linear algebra~\cite{JiangPW23}, formal languages~\cite{GibneyT21}, and database theory~\cite{BerkholzKS17,CaselS23,KaraNOZ23}.

We show the following relation to the NFA Acceptance hypothesis:

\begin{theorem} \label{thm:nfa_omv}
The NFA Acceptance hypothesis implies the OMv Hypothesis.
\end{theorem}

In particular, all implications of the OMv Hypothesis shown in~\cite{HenzingerKNS15,Dahlgaard16,AbboudD16,
HenzingerW21,CliffordGLS18,KempaK22,CliffordGK0U22,
DallantI22,LauR21,JiangPW23,GibneyT21,
BerkholzKS17,CaselS23,KaraNOZ23} also hold under the NFA Acceptance hypothesis.

\begin{proof}[Proof of Theorem~\ref{thm:nfa_omv}]
We reduce from Directed Edge-2-Colored Walk. To this end, we are given a directed graph $G=(V, E)$ with colors $c \colon E \to \{1,2\}$, designated vertices $s,t\in V$, and a color sequence $c_1, \dots, c_\ell \in \{1,2\}$. 
We assume the graph to be dense ($m = \Theta(n^2)$) and we assume $\ell = \Theta(n)$. 
By the NFA Acceptance hypothesis in the setting $\alpha=2, \beta=1$, there is no algorithm solving all such instances of Directed Edge-2-Colored Walk in time $O(n^{3-\eps})$ for any $\eps > 0$.

Let $N := \max\{n,\ell\}$. If $N>n$ we add $N-n$ isolated dummy vertices to $G$.

Consider for each color $c\in \{1,2\}$ the \emph{transposed} adjacency matrix $M^{(c)}$ corresponding to the edges with color $c$, i.e., $M^{(c)} \in \{0,1\}^{N \times N}$ where $M^{(c)}_{u,v} = 1$ if and only if $(v,u) \in E$ and $c(v,u) = c$. 
Let $u_0 \in \{0,1\}^N$ be the indicator vector for $s$, i.e., $u_0[v] = 1$ if and only if $v=s$.

Suppose there is an algorithm $\algo$ solving OMv in time $\Oh(n^{3-\varepsilon})$ for some $\varepsilon > 0$. We use $\algo$ to preprocess $M^{(1)}$ and $M^{(2)}$ (as independent OMv instances). For each $i=1, \dots, \ell$, we compute $u_i := M^{(c_i)} u_{i-1}$ using algorithm $\algo$ on the corresponding OMv instance. Finally, we accept the colored walk instance if and only if $u_\ell[t] = 1$.    

Inductively, it is straightforward to show that $u_i[v] = 1$ if and only if there is a walk from $s$ to $v$ with color sequence $c_1,\dots,c_i$; hence correctness follows. 
Note that for each of the two OMv instances we execute at most $\ell \le N$ rounds (and we can add dummy rounds to obtain exactly $N$ rounds). Therefore, algorithm $\algo$ solves both instances in total time $\Oh(N^{3-\varepsilon})$ over all rounds. Since $N = \max\{n,\ell\} = O(n)$ by our assumption on $\ell$, it follows that we can solve the given Directed Edge-2-Colored Walk instance in time $\Oh(n^{3-\varepsilon})$, contradicting the NFA Acceptance hypothesis (as discussed in the first paragraph). This finishes the proof.

\smallskip
\emph{Remark:} As described above, the reduction is not many-one, since we create two instances of OMv corresponding to the matrices $M^{(1)}$ and $M^{(2)}$. However, by a slight adaptation we can make the reduction many-one. To this end, we construct the matrix 
$$M = \begin{pmatrix} M^{(1)} & 0 \\ 0 & M^{(2)} \end{pmatrix},$$ 
which stacks the matrices $M^{(1)}, M^{(2)}$ in block matrices along the main diagonal. We use algorithm~$\algo$ to preprocess the OMv instance $M$. Then given a vector $u \in \{0,1\}^N$ and a color $c \in \{1,2\}$, we can compute the Boolean product $M^{(c)} u$ by one call to the OMv instance $M$: To this end, we set $u' := (u_1,\ldots,u_N,0,\ldots,0) \in \{0,1\}^{2N}$ if $c = 1$, and $u' := (0, \ldots, 0, u_1,\ldots,u_N) \in \{0,1\}^{2N}$ if $c = 2$, and we call $\algo$ to compute $M u'$. This yields $M^{(c)} u$, padded with some zeroes. Therefore, each step $u_i := M^{(c_i)} u_{i-1}$ performed in the above reduction can be implemented by one call to the OMv instance $M$. This makes the reduction many-one.
\end{proof}


\section{Equivalences of Colored Walk}
\label{sec:equivalences} 

In this section, we prove the equivalences of Colored Walk problems stated in Lemma~\ref{lem:equivalencehypo}. We also prove an equivalence to a version without source and target nodes, see \lemref{equivalence-to-anywalk} at the end of this section. Throughout this section we abbreviate Colored Walk as CW.

\RestateGo{\restateEquivalences}
\equivalences*

For the node version of the problem, we say that the \emph{(node) color sequence} of a walk $v_0,v_1,\ldots,v_k$ is the sequence $c(v_1),c(v_2),\ldots,c(v_k)$. Similarly, for the edge version we say that the \emph{(edge) color sequence} of a walk $v_0,v_1,\ldots,v_k$ is the sequence $c(v_0,v_1),c(v_1,v_2),\ldots,c(v_{k-1},v_k)$.

\medskip
The remainder of this section is devoted to the proof of \lemref{equivalencehypo} (and to the proof of \lemref{equivalence-to-anywalk}).
To this end, we define the following notion of reductions. Here, a \emph{parameter} is a function mapping any instance to a natural number, e.g., $n,m,\ell$ are parameters of the above problems.

\begin{definition}
  Let $X,Y$ be problems with the same set of parameters $P$. We say that there is a \emph{$P$\nobreakdash-preserving (many-one) reduction from $X$ to $Y$}, written $X \lePara Y$ if there is an algorithm $A$ that given an instance $I$ of $X$ computes an equivalent instance $J$ of $Y$ such that for all parameters $p \in P$ we have $p(J) \le p(I)^{1+o(1)}$, and $A$ runs in almost-linear time $N^{1+o(1)}$ in its input size $N$.
\end{definition}

The following lemmas show $\{n,m,\ell\}$-preserving reductions between the problems from \lemref{equivalencehypo}.

\begin{lemma}\label{lem:dir-n2-dir-e2}
Directed Node-2-CW {\lePara[\{n,m,\ell\}]} Directed Edge-2-CW.
\end{lemma}
\begin{proof}
Given a directed graph $G$ with node coloring $c \colon V \to \{1,2\}$, we define the edge coloring $c' \colon E \to \{1,2\}$ by setting $c'(u,v) := c(v)$ for all edges $(u,v) \in E$. This yields an equivalent Directed Edge-2-CW instance, since for any walk $v_0, v_1, \dots, v_\ell$ in $G$, the sequence of node colors $c(v_1),\dots,c(v_\ell)$ is the same as the sequence of edge colors $c'(v_0,v_1), \dots, c'(v_{\ell-1}, v_\ell)$. The reduction preserves the exact values of all parameters.
\end{proof}

\begin{lemma}\label{lem:dir-eo1-nfa}
Directed Edge-$n^{o(1)}$-CW {\lePara[\{n,m,\ell\}]} NFA Acceptance with $n^{o(1)}$ terminals.
\end{lemma}
\begin{proof}
This is a simple statement, as NFA Acceptance with $n^{o(1)}$ terminals is the generalization of Directed Edge-$n^{o(1)}$-CW where we allow loops and we allow multiple edges (with different labels) between two nodes. Since the two problems are formulated in a different language, we make the correspondence explicit:
Given a directed graph $G = (V,E)$ with designated vertices $s,t \in V$ and edge coloring $c \colon E \to \{1,\dots, C\}$ as well as a color sequence $c_1,\ldots,c_\ell \in \{1,\ldots,C\}$, we construct the NFA $M = (Q, \Sigma, \delta, q_0, F)$ by setting $Q := V$, $\Sigma := \{1,\ldots,C\}$, $\delta := \{(u,c(u,v),v) \mid (u,v) \in E\}$, $q_0 := s$, $F := \{t\}$, and we construct the string $S := c_1 c_2 \ldots c_\ell \in \Sigma^\ell$. It is straightforward to show that $M$ accepts $S$ if and only if there is a walk from $s$ to $t$ with color sequence $c_1 c_2 \ldots c_\ell$ in $G$. All parameters are preserved since $|Q|=|V|, |\delta|=|E|$ and $\ell$ remains unchanged.
\end{proof}

\begin{lemma}\label{lem:nfa-dir-eo1}
NFA Acceptance with $n^{o(1)}$ terminals {\lePara[\{n,m,\ell\}]} Directed Node-$n^{o(1)}$-CW.
\end{lemma}
\begin{proof}
Given an NFA $M = (Q, \Sigma, \delta, q_0, F)$ with set of states $Q$, input alphabet $\Sigma$, transitions $\delta \subseteq Q \times \Sigma \times Q$, initial state $q_0$, and set of accepting states $F \subseteq Q$, as well as a string $S \in \Sigma^\ell$, we first transform it into an equivalent instance without loops and with exactly one accepting state. 

To remove loops, we replace the states $Q$ by two copies: $\hat Q := Q \times \{1,2\}$. Each transition $(q,\sigma,q')$ is replaced by two transitions $((q,1),\sigma,(q',2))$ and $((q,2),\sigma,(q',1))$; note that this ensures that we have no loops (i.e., no transition starts and ends in the same state). We also replace $q_0$ by $\hat q_0 := (q_0,1)$ and $F$ by $\hat F := F \times \{1,2\}$. It is easy to check that the new NFA accepts $S$ if and only if the old NFA accepts $S$, and the size of the NFA is only changed by a constant factor.

To ensure exactly one accepting state, we add a new state $f_0$ to $\hat Q$. We fix an arbitrary symbol $\sigma \in \Sigma$, and we add transitions $(f,\sigma,f_0)$ for all $f \in \hat F$. Finally, we replace $\hat F$ by $\{f_0\}$ and $S$ by $S \sigma$. Note that this results in an equivalent instance with exactly one accepting state and all sizes stay the same up to constant factors. Therefore, in the following we assume without loss of generality that the given NFA $M = (Q, \Sigma, \delta, q_0, F)$ satisfies $F = \{f_0\}$ and has no loops. Moreover, we can assume without loss of generality that $\Sigma = \{1,\ldots,|\Sigma|\}$.

We define a graph $G=(V, E)$ with vertex set $V= Q \times \Sigma$. For each transition $(q,\sigma',q') \in \delta$, we add the edges $(\,(q,\sigma)\,,\,(q',\sigma')\,) \in E$ for all $\sigma \in \Sigma$. We define the vertex coloring $c \colon V \to \{1,\ldots,|\Sigma|\}$ by $c(\, (q,\sigma) \,) = \sigma$. Observe that for any choice of ``starting color'' $\sigma_0 \in \{1,\dots, |\Sigma|\}$, the mapping from any transition sequence $q_0 \stackrel{\sigma_1}{\to} q_1 \stackrel{\sigma_2}{\to} q_2 \dots \stackrel{\sigma_\ell}{\to} q_\ell$ in $M$ to the walk $v_0 := (q_0,\sigma_0)\, , \, v_1:= (q_1, \sigma_1)\, , \dots, \, v_\ell := (q_\ell, \sigma_\ell)$ in $G$ is indeed a one-to-one mapping from transition sequences in $M$ to walks in $G$ starting in the node $(q_0,\sigma_0)$ such that the string read by the transition sequence equals the color sequence of the walk. In particular, $M$ accepts the string $S$ if and only if there is a walk from $s := (q_0,\sigma_0)$ to $t := (f_0,S[\ell])$ in $G$ with color sequence $S[1],\ldots,S[\ell]$.
Thus, for any NFA~$M$ we can construct, in linear time in the output, an equivalent Directed Node-$|\Sigma|$-CW instance.

We verify that the parameters are preserved for $|\Sigma|=|Q|^{o(1)}$: The number of vertices in $G$ is bounded by $|V| = |Q|\cdot |\Sigma| = |Q|^{1+o(1)}$, similarly, we have $|E| \le |\delta|\cdot |\Sigma| = |\delta|^{1+o(1)}$. Finally, the length $\ell$ of the color sequence equals the length of the string $S$.
\end{proof}

\begin{lemma}\label{lem:dir-nn-dir-n2}
Directed Node-$n$-CW {\lePara[\{n,m,\ell\}]} Directed Node-2-CW.
\end{lemma}
\begin{proof}
Given a directed graph $G$ with node coloring $c \colon V \to \{1,\dots, C\}$, we let $B := \lceil \log_2 C \rceil$ and define a graph $G' = (V', E')$ with vertex set $V' = V \times [B]$. For every edge $(u,v)\in E$, we add the edge $(\, (u,B)\, , \, (v, 1)\,)\in E'$, and for every $v\in V$ and $1 \le i < B$, we add the ``path'' edge $(\, (v, i)\,, \, (v, i+1) \,) \in E'$. For a color $c\in \{1,\dots, C\}$, let $\bin(c, i)$ be the $i$-th bit in the $B$-bit binary representation of $c$. We define the node coloring $c'\colon V' \to \{1,2\}$ by $c'(\, (v, i)\,) = \bin(c(v), i) + 1$ for all $v\in V, i \in [B]$. Thus, for any Directed Node-$n$-CW instance $G, s, t, (c_1, \dots, c_\ell)$, we define a corresponding instance on the graph $G'$ with source $(s,B)$, target $(t, B)$ and color sequence $c'_1, \dots, c'_{\ell B}$ where we set $c'_{(j-1)B + i} := \bin(c_j, i)+1$ for all $1\le j \le \ell$ and $1\le i \le B$. It is straightforward to verify that $v_0, v_1, \dots, v_\ell$ is a walk in $G$ with color sequence $c_1 = c(v_1), \dots, c_\ell = c(v_\ell)$ if and only if $v'_0, v'_1, \dots, v'_{\ell B}$ with $v'_0 = (v_0, B)$ and $v'_{(j-1)B + i} = (v_j, i)$, $1\le j\le \ell, 1\le i \le B$ is a walk in $G'$ with color sequence $c'_1 = c'(v'_1), \dots, c'_{\ell B} = c'(v'_{\ell B})$. 

Note that for $C=n$, we have $B = \Oh(\log n)$ and thus $|V'| = B|V| \le |V|^{1+o(1)}$, $|E'| = |E|+|V|B \le |E|^{1+o(1)}$ and $\ell B \le \ell^{1+o(1)}$.
\end{proof}

To reduce from Node-2-CW in a \emph{directed} graph $G$ to Edge-2-CW in an \emph{undirected} graph $G'$, a natural attempt is the following: each vertex $v$ is replaced by two vertices, $\vin$ and $\vout$, such that a directed edge $(u,v)$ can be represented by the undirected edge $\{\uout, \vin\}$. Additionally, we introduce the edges $\{\vin, \vout\}$ of color $c(v)$. If we were allowed to introduce an additional color $\cspecial \notin \{1,2\}$, we could label each edge $\{\uout, \vin\}$ with $\cspecial$, and any walks with node color sequence $c_1, c_2, \dots, c_\ell$ in $G$ would uniquely correspond to walks with edge colors sequence in $\cspecial,  c_1,  \cspecial, c_2,  \dots \cspecial,  c_\ell$ and vice versa. To avoid the blow-up in the number of colors, however, we must reuse a color for $\cspecial$, say $\cspecial = 1$. Since this would allow a walk to use an edge $\{\uin, \vout\}$ in the ``reverse'' direction towards $\uin$ whenever we are supposed to ``check'' that a node has color 1, we might obtain illegal transitions in the resulting walk. An analogous reasoning applies for a reduction to undirected Node-2-CW.

We use slightly more involved gadgetry for both reductions: For a given directed graph $G= (V,E)$ and an integer $P$, we create an \emph{undirected} graph $G' = (V', E')$ with vertex set $V' = \{ \vin, \vout, v^{(1)}, \dots, v^{(P)} \mid v\in V\}$ as follows: we introduce, for each edge $(u,v) \in E$, the edges $\{\uout, \vin\}$ in $E'$, as well as, for all $v\in V$, all ``path edges'' $\{\vin, v^{(1)}\}, \{v^{(P)}, \vout\}$ and $\{v^{(i)}, v^{(i+1)}\}$ for all $1 \le i < P$. By choosing the color sequences for all path vertices/edges appropriately, we will be able to enforce that every edge is used in ``forward'' direction, i.e., any feasible walk traverses an edge $\{\uout, \vin\}$, then all path edges toward $\vout$, then an edge $\{\vout, w_\mathrm{in}\}$, and so on.

\begin{lemma}\label{lem:dir-n2-undir-e2}
Directed Node-2-CW {\lePara[\{n,m,\ell\}]} Undirected Edge-2-CW.
\end{lemma}
\begin{proof}
Given a directed graph $G=(V,E)$, we construct the undirected version $G'$ using path length $P = 4$ as described above. We define the coloring $c' \colon E \to \{1,2\}$ as follows: for all $(u,v)\in E$, we set $c'(\uout, \vin) = 1$. For each $v\in V$, we set the colors of the path edges to  
\[ (c'(\vin, v^{(1)})\,,\, c'(v^{(1)}, v^{(2)})\,,\, c'(v^{(2)}, v^{(3)}) \, , \, c'(v^{(3)}, v^{(4)})\, , \, c'( v^{(4)}, \vout )) := (2 \, , \, c(v) \, , \, c(v) \, , \, 1 \, , \, 2). \]
The main property of this construction is captured by the following claim.
In the remainder of the proof, for any $c \in \{1,2\}$ and $v \in V$ we set:
\[ \colors(c) := (2, c, c, 1, 2) \qquad\text{and}\qquad \pathseq(v) := (\vin, v^{(1)}, v^{(2)}, v^{(3)}, v^{(4)}, \vout). \]
In what follows, for any walk $z=(z_0,\dots, z_\ell)$ with color sequence $c'= (c'_1, \dots, c'_\ell)$, we use the shorthands $z_{\dots j}= (z_0,\dots, z_j)$ and $c'_{\dots j} = (c'_1, \dots, c'_j)$ to denote the length-$j$ prefix of $z$ and its corresponding color sequence, respectively.
\begin{claim}
Let $v\in V$ and $c\in \{1,2\}$. There is a walk $z= (z_0, z_1, \dots, z_5)$ in $G'$ with $z_0 = \vin$ and edge color sequence $\colors(c)$ if and only if $c = c(v)$ and $z = \pathseq(v)$.
\end{claim}
\begin{proof}
If $c=c(v)$, then the walk $\pathseq(v)$ clearly has edge color sequence $(2, c, c, 1, 2)$. For the converse, let $z = (z_0,\dots, z_5)$ be a walk with $z_0 = \vin$ and color sequence $\colors(c)$. We analyze the possible values for $z_{\dots i}$ when observing the color sequence $\colors(c)_{\dots i}$ for $i=1, \dots, 5$:
\begin{itemize}
\item $z_1 = v^{(1)}$ and thus $z_{\dots 1} = \pathseq(v)_{\dots 1}$ always holds, as $\{\vin, v^{(1)}\}$ is the only edge adjacent to $\vin$ that has color 2, because the edges $\{\uout,\vin\}$ have color 1.
\item Consider the case $c=1$:
\begin{itemize}
\item If $c(v) \ne c = 1$, then there is no edge leaving $z_1 = v^{(1)}$ with color $c = 1$. Thus we may assume in the remainder of the case that $c(v)=c=1$. 
\item It follows that $z_2 = v^{(2)}$ and hence $z_{\dots 2} = \pathseq(v)_{\dots 2}$.
\item Since both edges adjacent to $z_2=v^{(2)}$ have color $c(v)=c=1$, we have either $z_{\dots 3} = \pathseq(v)_{\dots 3}$ or $z_{3} = v^{(1)}$.
\item Since the only nodes adjacent to $v^{(3)}$ or $v^{(1)}$ via an edge with color 1 are $v^{(4)}$ and $v^{(2)}$, we have that $z_{4} \in \{v^{(4)},v^{(2)}\}$, where $z_4 = v^{(4)}$ only occurs if $z_{\dots 4} = \pathseq(v)_{\dots 4}$. 
\item Finally, as desired, $z_5 = \vout$ must hold, since among $\{v^{(4)}, v^{(2)}\}$, only $v^{(4)}$ has an adjacent edge of color 2, which leads to $\vout$. This yields $z = \pathseq(v)$ for $c=1$.
\end{itemize}
\item We analyze the remaining case $c=2$:
\begin{itemize}
\item We have $z_2 = \vin$, or, only if $c(v)=c=2$, possibly $z_2 = v^{(2)}$ and hence $z_{\dots 2} = \pathseq(v)_{\dots 2}$.
\item Likewise, since $\{\vin, v^{(1)}\}$ is the only edge adjacent to $\vin$ with color $c=2$, we must have $z_3 = v^{(1)}$, or, only if $c(v)=c=2$ and $z_{\dots 2} = \pathseq(v)_{\dots 2}$, possibly $z_3 = v^{(3)}$ (and hence $z_{\dots 3} = \pathseq(v)_{\dots 3}$).
\item Note that $v^{(1)}$ only has an adjacent edge of color 1 if $c(v)=1\ne c$. In contrast, if $c(v) = c= 2$, $v^{(3)}$ has its only adjacent edge of color 1 to $v^{(4)}$. Thus, we have either $z_4 = v^{(2)}$, which can only happen if $c(v) \ne c$, or $z_4 = v^{(4)}$, which can only happen if $c(v) = c$ and $z_{\dots 3} = \pathseq(v)_{\dots 3}$. In the latter case, we thus must have $z_{\dots 4} = \pathseq(v)_{\dots 4}$.
\item Finally, $z_5 = \vout$, $c(v)=c$ and $z_{\dots 4}=\pathseq(v)_{\dots 4}$ must hold, since in the case $c\ne c(v)$, $v^{(2)}$ has no adjacent edges of color 2, while $v^{(4)}$ has an edge of color 2 to $\vout$. Thus $z = \pathseq(v)$ and $c(v) = c$.
\end{itemize}
\end{itemize}
\end{proof}
Given a Directed Node-2-CW instance $G$, $s,t\in V$ and $c_1, \dots, c_\ell$, we construct, in linear time in the output, an Undirected Edge-2-CW instance $G'$ with source $s_\mathrm{out}$, target $t_\mathrm{out}$ and color sequence $1, \colors(c_1), 1, \colors(c_2), \dots, 1, \colors(c_\ell)$. For any walk $s = v_0, v_1, \dots, v_\ell =t$ with color sequence $c_1,\ldots,c_\ell$ in $G$, the walk $v_0, \pathseq(v_1), \pathseq(v_2), \dots, \pathseq(v_\ell)$ in $G'$ has the desired color sequence $1, \colors(c_1), 1, \colors(c_2), \dots, 1, \colors(c_\ell)$ by the above claim. Conversely, we see that for each color substring $1, \colors(c_i)$ any walk in $G'$ that starts in some node $\uout$ is of the form $(\uout, \pathseq(v))$ for some $v\in V$ with $(u,v)\in E$ and $c(v)=c_i$, since the only 1-colored edges adjacent to $\uout$ lead to some $\vin$ with $(u,v)\in E$ and the above claim proves that the walk $\pathseq(v)$ must follow, which requires $c(v)=c_i$. By repeated application of this fact, any walk in $G'$ with color sequence $1, \colors(c_1), 1, \colors(c_2), \dots, 1, \colors(c_\ell)$ corresponds to a walk in $G$ with color sequence $c_1,\dots, c_\ell$, as desired. Note that all parameters have increased by at most a constant factor, which yields the desired reduction.
\end{proof}

For a reduction to Undirected \emph{Node}-2-CW, we imitate the above reduction by defining a suitable node color sequence for path nodes. The analysis is slightly simpler than for Edge-2-CW.

\begin{lemma}\label{lem:dir-n2-undir-n2}
Directed Node-2-CW {\lePara[\{n,m,\ell\}]} Undirected Node-2-CW.
\end{lemma}
\begin{proof}
Given a directed graph $G=(V,E)$, we construct the undirected version $G'$ using path length $P = 4$ as described above. We define the coloring $c' \colon V' \to \{1,2\}$, as follows: for all $v\in V$, we set $c'(\vin) = c'(\vout) = 1$ and set
\[ (c'(v^{(1)})\,,\, c'(v^{(2)})\,,\, c'(v^{(3)}) \, , \, c'(v^{(4)})) := (2 \, , \, c(v) \, , \, 2 \, , \, 2). \]
The main property of this construction is captured by the following claim. In the remainder of the proof, for any $c \in \{1,2\}$ and $v \in V$ we set:
\[ \colors(c) := (2, c, 2, 2, 1) \qquad\text{and}\qquad \pathseq(v) := (\vin, v^{(1)}, v^{(2)}, v^{(3)}, v^{(4)}, \vout). \]

Recall that for any walk $z=(z_0,\dots, z_\ell)$ with color sequence $c'= (c'_1, \dots, c'_\ell)$, we use the shorthands $z_{\dots j}= (z_0,\dots, z_j)$ and $c'_{\dots j} = (c'_1, \dots, c'_j)$ to denote the length-$j$ prefix of $z$ and its corresponding color sequence, respectively.
\begin{claim}
Let $v\in V$ and $c \in \{1,2\}$. There is a walk $z = (z_0, z_1, \dots, z_5)$ with $z_0 = \vin$ and color sequence $\colors(c)$ in $G'$ if and only if $c = c(v)$ and $z = \pathseq(v)$.
\end{claim}
\begin{proof}
If $c=c(v)$, then the walk $\pathseq(v)$ clearly has node-color sequence $(2, c, 2, 2, 1)$. For the converse, let $z=(z_0,\dots, z_5)$ be a walk with $z_0 =\vin$ and color sequence $\colors(c)$. We analyze the possible values for $z_{\dots i}$ when observing the color sequence $\colors(c)_{\dots i}$ for $i=1,\dots 5$:
\begin{itemize}
\item We have $z_1 = v^{(1)}$ and thus $z_{\dots 1} = \pathseq(v)_{\dots 1}$, since $v^{(1)}$ is the only neighbor of $\vin$ with color 2.
\item We have $z_2 = \vin$ (which might happen if $c=1$), or $z_2 = v^{(2)}$ (which can happen if and only if $z_{\dots 1} = \pathseq(v)_{\dots 1}$ and $c=c(v)$). In the latter case, this yields $z_{\dots 2} = \pathseq(v)_{\dots 2}$ and $c=c(v)$.
\item Since the only neighbor of $\vin$ with color 2 is $v^{(1)}$, we have either that $z_3 = v^{(1)}$ or, only if $z_{\dots 2} = \pathseq(v)_{\dots 2}$ and $c=c(v)$, that $z_3=v^{(3)}$. In the latter case, it holds that $z_{\dots 3} = \pathseq(v)_{\dots 3}$ and $c=c(v)$.
\item It follows that $z_4 = v^{(2)}$ (which might happen if $c(v)=2$), or $z_4=v^{(4)}$ (which can happen only if $z_{\dots 3} = \pathseq(v)_{\dots 3}$ and $c=c(v)$). In the latter case, it holds that $z_{\dots 4} = \pathseq(v)_{\dots 4}$ and $c=c(v)$.
\item Finally, since $v^{(2)}$ has no neighbor of color 1, we have $z_5=v^{(5)}$, which can happen if and only if $z_{\dots 4} = \pathseq(v)_{\dots 4}$ and $c=c(v)$. This yields $z= \pathseq(v)$ and $c=c(v)$, as desired.
\end{itemize}
\end{proof}
The remainder of the reduction is analogous to the proof of Lemma~\ref{lem:dir-n2-undir-e2}: Given a Directed Node-2-CW instance $G$, $s,t\in V$ and color sequence $c_1, \dots, c_\ell$, we construct, in linear time in the output, an Undirected Node-2-CW instance $G'$ with source $s_\mathrm{out}$, target $t_\mathrm{out}$ and color sequence $1, \colors(c_1), 1, \colors(c_2), \dots, 1, \colors(c_\ell)$. For any walk $s=v_0, v_1, \dots, v_\ell =t$ with color sequence $c_1,\ldots,c_\ell$ in $G$, the walk $v_0, \pathseq(v_1), \pathseq(v_2), \dots, \pathseq(v_\ell)$ in $G'$ has the desired color sequence $1, \colors(c_1), 1, \colors(c_2), \dots, 1, \colors(c_\ell)$ by the above claim. Conversely, we see that for each color substring $1, \colors(c_i)$ any walk in $G'$ that starts in some $\uout$ is of the form $(\uout, \pathseq(v))$ for some $v\in V$ with $(u,v)\in E$ and $c(v)=c_i$, since the only 1-colored nodes adjacent to $\uout$ lead to some $\vin$ with $(u,v)\in E$ and the above claim proves that the walk $\pathseq(v)$ must follow, which requires $c(v)=c_i$. By repeated application of this fact, any walk in $G'$ with color sequence $1, \colors(c_1), 1, \colors(c_2), \dots, 1, \colors(c_\ell)$ corresponds to a walk in $G$ with color sequence $c_1,\dots, c_\ell$, as desired. Note that all parameters have increased by at most a constant factor, which yields the desired reduction.
\end{proof}

Now we combine the above relations to a full equivalence.
\begin{lemma} \label{lem:equivalences}
  The following problems are equivalent under $\{n,m,\ell\}$-preserving reductions:
  \begin{itemize}
    \item NFA Acceptance with $n^{o(1)}$ terminals,
    \item Directed Node-2-Colored Walk,
    \item Directed Node-$n$-Colored Walk,
    \item Directed Edge-2-Colored Walk,
    \item Directed Edge-$n^{o(1)}$-Colored Walk,
    \item Undirected Node-2-Colored Walk,
    \item Undirected Node-$n$-Colored Walk,
    \item Undirected Edge-2-Colored Walk,
    \item Undirected Edge-$n^{o(1)}$-Colored Walk.
  \end{itemize}
\end{lemma}

\begin{proof}
The parameter-preserving equivalence of all directed variants for parameters $P=\{n,m,\ell\}$ follows from the following chain of reductions: 
\begin{align*}
\text{Directed Node-2-CW} 	& \lePara \text{Directed Edge-2-CW} & & \text{(by \lemref{dir-n2-dir-e2})}\\
         			& \lePara \text{Directed Edge-$n^{o(1)}$-CW} & & \text{(trivially by definition)}\\
         			& \lePara \text{NFA Acceptance with $n^{o(1)}$ terminals} & & \text{(by \lemref{dir-eo1-nfa})},\\
			       	& \lePara \text{Directed Node-$n^{o(1)}$-CW}& & \text{(by \lemref{nfa-dir-eo1})},\\
				& \lePara \text{Directed Node-$n$-CW} & &  \text{(trivially by definition)},\\
			        & \lePara \text{Directed Node-2-CW} & &  \text{(by \lemref{dir-nn-dir-n2})}.
\end{align*}

Note that for any $C$, Undirected Node-$C$-CW $\lePara$ Directed Node-$C$-CW and Undirected Edge-$C$-CW $\lePara$ Directed Edge-$C$-CW follows trivially by replacing each undirected edge by directed edges in both directions, while reductions in the other direction (for $C=2$) are given by \lemrefs{dir-n2-undir-e2}{dir-n2-undir-n2}. We use these to also conclude equivalence to the undirected variants by
\begin{align*}
\text{Undirected Node-2-CW} 	& \lePara \text{Undirected Node-$n$-CW} & & \text{(trivially by definition)}\\
         			& \lePara \text{Directed Node-$n$-CW} & & \text{(trivially, as argued above)}\\
			       	& \lePara \text{Directed Node-2-CW}& & \text{(by already established equivalence)},\\
				& \lePara \text{Undirected Node-2-CW} & &  \text{(by \lemref{dir-n2-undir-n2})},
\end{align*}
and similarly,
\begin{align*}
\text{Undirected Edge-2-CW} 	& \lePara \text{Undirected Edge-$n^{o(1)}$-CW} & & \text{(trivially by definition)}\\
         			& \lePara \text{Directed Edge-$n^{o(1)}$-CW} & & \text{(trivially, as argued above)}\\
			       	& \lePara \text{Directed Node-2-CW}& & \text{(by already established equivalence)},\\
				& \lePara \text{Undirected Edge-2-CW} & &  \text{(by \lemref{dir-n2-undir-e2})}.
\end{align*}
\end{proof}

\begin{lemma} \label{lem:preservinghypo}
Let $\alpha \in [1,2], \beta > 0$ and $c \ge 1$. Let $X,Y$ be any problems listed in Lemma~\ref{lem:equivalences}. If problem~$X$ restricted to instances with $m = O(n^\alpha)$ and $\ell = O(n^\beta)$ has no $O(n^{c-\eps})$-time algorithm for any $\eps>0$, then the same holds for problem $Y$.
\end{lemma}
\begin{proof}
Suppose for the sake of contradiction that problem $Y$ restricted to instances with $m = O(n^\alpha)$ and $\ell = O(n^\beta)$ has an algorithm $\algo$ running in time $O(n^{c-\eps})$ for some $\eps > 0$. Then we can solve problem $X$ as follows. Given an instance $I$ of problem $X$ with parameters $n,m,\ell$ satisfying $m = O(n^\alpha), \ell = O(n^\beta)$, run the $\{n,m,\ell\}$-preserving reduction to obtain an equivalent instance $I_0$ of $Y$ with parameters $n_0, m_0, \ell_0$. Recall that $n_0 \le n^{1+o(1)}, m_0 \le m^{1+o(1)}, \ell_0 \le \ell^{1+o(1)}$.
We set $n_1 := \max\{n_0, m_0^{1/\alpha}, \ell_0^{1/\beta}\}$. We add $n_1-n_0$ isolated nodes/states to $I_0$, resulting in an instance $I_1$ with parameters $n_1, m_1 = m_0, \ell_1 = \ell_0$. Note that $m_1 \le n_1^\alpha$ and $\ell_1 \le n_1^\beta$, so we can run algorithm $\algo$ on instance $I_1$. Since $I_1$ is equivalent to $I_0$, and thus equivalent to $I$, this solves the given instance $I$ of problem $X$. 

Tracing the above inequalities, we observe that $n_1 \le n^{1+o(1)}$, so the running time of algorithm~$\algo$ on instance $I_1$ is $O(n_1^{c-\eps}) = n^{c-\eps+o(1)}$. The reduction runs in almost-linear time in the input size, i.e., in time $(n+m+\ell)^{1+o(1)} = n^{\max\{\alpha,\beta\}+o(1)}$. Hence, the total time to solve instance $I$ is $O(n^{c-\eps+o(1)} + n^{\max\{\alpha,\beta\}+o(1)})$. If $c \le \max\{\alpha,\beta\}$ then the conclusion that $Y$ has no $O(n^{c-\eps})$-time algorithm is trivial, so we can assume $c > \max\{\alpha,\beta\}$.
After possibly replacing $\eps$ by $\min\{\eps, c - \max\{\alpha,\beta\}\}$, we can further bound the running time by $O(n^{c-\eps+o(1)})$.
Bounding $n^{o(1)} \le O(n^{\eps/2})$, the running time becomes $O(n^{c-\eps/2}) = O(n^{c-\eps'})$ for $\eps' := \eps/2$. This contradicts the assumption that problem $X$ restricted to instances with $m = O(n^\alpha), \ell = O(n^\beta)$ has no $O(n^{c-\eps''})$-time algorithm for any $\eps''>0$. 
\end{proof}

Finally we are ready to prove \lemref{equivalencehypo}.

\begin{proof}[Proof of \lemref{equivalencehypo}]
Let $X,Y$ be any problems listed in the lemma statement. If the lemma claim corresponding to $X$ holds, then in particular problem $X$ restricted to instances with $m = O(n^\alpha)$ and $\ell = O(n^\beta)$ has no $O(n^{\alpha+\beta-\eps})$-time algorithm for any $\eps > 0$. Using the $\{n,m,\ell\}$-preserving reduction from $X$ to $Y$ guaranteed by \lemref{equivalences}, by \lemref{preservinghypo} we obtain the same result for $Y$, i.e., problem~$Y$ restricted to instances with $m = O(n^\alpha)$ and $\ell = O(n^\beta)$ has no $O(n^{\alpha+\beta-\eps})$-time algorithm for any $\eps > 0$. To arrive at the lemma claim corresponding to $Y$, we need to additionally ensure $m = \Omega(n^\alpha)$ (and for some $Y$ also $\ell = \Omega(n^\beta)$), which we achieve by padding, as follows.

To pad the parameter $m$, we add $n$ dummy nodes (or states) and among them we add $\Theta(n^\alpha)$ dummy edges (or transitions), to ensure $m = \Theta(n^\alpha)$. Since these new nodes and edges are disconnected from the old graph, they do not change the result.

For the directed variants of Colored Walk we can also pad the length $\ell$: We add new nodes $s_0,s_1$ and edges $(s_0,s_1)$ and $(s_1,s_0)$ of color 1 and an edge $(s_1,s)$ of color 2, and we change the color sequence to $1^{2 \lceil n^\beta \rceil} \, 2 \, c_1 \ldots c_\ell$. Observe that there is a walk from $s_0$ to $t$ with the new color sequence in the new graph if and only if there is a walk from $s$ to $t$ with the old color sequence $c_1 \ldots c_\ell$ in the old graph. An analogous construction can be used to pad the number of transitions in a given NFA. Therefore, we can assume $\ell = \Theta(n^\beta)$ for the directed variants of Colored Walk as well as for NFA Acceptance.

Since these constructions change $n$ only by a constant factor, we still rule out a running time of $O(n^{\alpha+\beta-\eps})$. This finishes the proof.
\end{proof}

Finally, we also show that Colored Walk is equivalent to a formulation without given source and target nodes $s,t\in V$, specifically, the variant in which we are looking for \emph{any} walk in $G$ with the prescribed color sequence $c_1,\dots, c_\ell$. We call these formulations \textsc{Directed/Undirected $\Sigma$-Edge-Colored/$\Sigma$-Node-Colored AnyWalk}.

Below, we prove a $\{n,m,\ell\}$-preserving equivalence of Undirected Edge-Colored Walk and Undirected Edge-Colored AnyWalk (up to an additive constant in the alphabet size); the proof of equivalence for all directed and/or node-colored variants is analogous (and in fact simpler).

\begin{lemma}[Equivalence of $s$-$t$-version and any-walk version]\label{lem:equivalence-to-anywalk} The following relationships hold:
\begin{enumerate}
	\item Undirected Edge-2-Colored Walk {\lePara[\{n,m,\ell\}]} Undirected Edge-4-Colored AnyWalk, \label{item:walk-to-anywalk}
	\item Undirected Edge-$\sigma$-Colored AnyWalk {\lePara[\{n,m,\ell\}]} Undirected Edge-$(\sigma+1)$-Colored Walk. \label{item:anywalk-to-walk}
\end{enumerate}	
\end{lemma}
\begin{proof}
	For~\ref{item:walk-to-anywalk}, consider any Undirected Edge-2-Colored Walk instance $G=(V,E)$, $s,t\in V$ and $c_1,\dots, c_\ell \in \{1,2\}$. We construct a graph $G'$ from $G$ by adding two nodes $s',t'$, connecting $s'$ to $s$ by an edge labeled 3 and connecting $t$ to $t'$ by an edge labeled 4. Note that any walk in $G'$ with color sequence $3,c_1,\dots,c_\ell,4$ must start in $s'$, transition to $s$, use a walk from $s$ to $t$ in $G$ with color sequence $c_1,\dots, c_\ell$, and finally transition to $t'$. Thus, $G'$ with color sequence $3,c_1,\dots,c_\ell,4$ yields an equivalent Undirected 4-Edge-Colored AnyWalk instance. The reduction is clearly $\{n,m,\ell\}$-preserving.
	
	For~\ref{item:anywalk-to-walk}, consider any Undirected Edge-2-Colored AnyWalk instance $G=(V,E)$, $c_1,\dots, c_\ell \in [\sigma]$. We construct $G'$ from $G$ by adding two nodes $s',t'\in V$, connecting $s'$ to each node $v\in V$ by an edge labeled $\sigma + 1$, as well as connecting each node $v\in V$ to $t'$ by an edge labeled $\sigma + 1$. Note that there is a walk from $s'$ to $t'$ in $G$ with color sequence $(\sigma+1),c_1,\dots, c_\ell, (\sigma+1)$ if and only if there is any walk in $G$ with color sequence $c_1,\dots, c_\ell$. Thus, we have obtained an equivalent Undirected Edge-$(\sigma+1)$-Colored Walk instance. Since we have added only two nodes, $O(n)=O(m)$ edges and two colors to the color sequence, the reduction is $\{n,m,\ell\}$-preserving.
\end{proof}

We obtain in particular that the NFA Acceptance hypothesis is equivalent to the $(m\ell)^{1-o(1)}$ barrier for Directed/Undirected Edge-4-Colored/Node-4-Colored AnyWalk.

\section{Relation to Other Hypotheses}
\label{sec:otherHypos}

For the reader's convenience, we restate the known connections between the NFA Acceptance hypothesis and standard fine-grained hypotheses, as described in Section~\ref{sec:intro}. 

\subsection{Tight Lower Bound for Sparse NFAs Under SETH}

In the Orthogonal Vectors (OV) problem, we are given vector sets $A, B\subseteq \{0,1\}^d$, and the task is to determine whether there is an \emph{orthogonal pair} $(a,b)\in A \times B$, i.e., for all $1 \le k\le d$ we have $a[k] = 0$ or $b[k]=0$.

\begin{hypothesis}[OV Hypothesis]
	For no $\epsilon > 0$ and $\beta > 0$, there is an algorithm solving OV with sets of size $n \coloneqq  |A|$ and $\ell \coloneqq |B| =\Theta(n^\beta)$ and $d\le n^{o(1)}$ dimensions in time $O((n\ell)^{1-\epsilon}) = O(n^{(1+\beta)(1-\epsilon)})$.
\end{hypothesis}
This hypothesis is usually stated with sets of the same size, but is equivalent to the above version, see, e.g.,~\cite[Lemma II.1]{BringmannK15}. The hypothesis is well known to be implied by the Strong Exponential Time Hypothesis (SETH)~\cite{Williams05}.

The following conditional lower bound can be found for different versions of the Colored Walk problem throughout the algorithmic literature, including~\cite{BackursI16,BringmannGL17, EquiMTG23, PotechinS20, CaselS23} -- recall that these turn out to be fine-grained equivalent due to Lemma~\ref{lem:equivalencehypo}.

\begin{proposition}\label{prop:sparseLB}
If there are $\epsilon,\beta > 0$ such that NFA Acceptance on an $n$-state NFA with $O(n)$ transitions and a string of length $\ell=\Theta(n^\beta)$ over $\Sigma=\{0,1,2\}$ can be solved in time $O( (\ell n)^{1-\epsilon})$, then the OV Hypothesis fails.
\end{proposition}

Note that this establishes a tight conditional lower bound for NFAs that are \emph{sparse}. 

\begin{proof}
We give a $\{n,\ell\}$-preserving reduction from OV to NFA Acceptance. It is not difficult to see that the claimed lower bound follows, analogously to Lemma~\ref{lem:preservinghypo}.

Given an OV instance $A,B\subseteq\{0,1\}^d$ with $n=|A|$, $\ell=|B|$, $d=n^{o(1)}$, we construct an NFA $M$ over $\Sigma=\{0,1,2\}$ as follows: $M$ has initial state $s$, accepting state $t$. We write $A=\{a_1,\dots, a_n\}$ and introduce, for every $1\le i \le n$, the states $q^{(i)}_0,\dots, q^{(i)}_d$. For every $1\le i \le n$ and $1\le k\le d$, we add the transitions $(q^{(i)}_{k-1},0,q^{(i)}_k)$ and, if $a_i[k]=0$, the transition $(q^{(i)}_{k-1},1,q^{(i)}_k)$.
Concluding the construction of $M$, we add transitions $(s,2,q^{(i)}_0)$ and $(q^{(i)}_d,2,t)$ for all $1\le i\le n$ and the loop transitions $(q,0,q),(q,1,q),(q,2,q)$ for each $q\in \{s,t\}$.

To construct the string $S$, we write $B=\{b_1,\dots, b_\ell\}$ and define $S \coloneqq 2 \; b_1 \; 2 \; \dots \; 2 \; b_\ell \; 2$, where we interpret $b_i$ as a length-$d$ string over $\{0,1\}$.

We argue that $M$ accepts $S$ if and only if $A,B$ contain an orthogonal pair: By the structure of $M$ and $S$, the only accepting runs of $M$ have the following form. $M$ reads a prefix $2 b_1 2 \dots b_{j-1}$ for some $1\le j \le \ell$ while staying in $s$, then uses the transition sequence 
\begin{equation} \label{eq:transitions}
s \stackrel{2}{\to} q^{(i)}_0 \stackrel{b_j[1]}{\to} q^{(i)}_1 \stackrel{b_j[2]}{\to}\cdots \stackrel{b_j[d]}{\to} q^{(i)}_d \stackrel{2}{\to} t, \end{equation}
for some $1\le i \le n$, and then stays in $t$ while reading the suffix $b_{j+1} 2 \dots b_\ell 2$. By observing that the transition sequence~\eqref{eq:transitions} exists in $M$ if and only if $a_i, b_j$ are orthogonal, the claim follows.

Finally, we note that this indeed yields a $\{n,\ell\}$-preserving reduction: $M$ has $O(nd)=n^{1+o(1)}$ states, and $S$ has length $O(\ell d)=\ell^{1+o(1)}$, and they can be computed in linear time in the input.
\end{proof}

\subsection{Tight Combinatorial Lower Bound from Triangle and \texorpdfstring{\boldmath{$k$}-Clique}{k-Clique}}

In the $k$-Clique problem, we are given an undirected, simple graph $G=(V,E)$, and the task is to determine whether there is a set $C \subseteq V$ of $k$ nodes that are pairwise adjacent, i.e., $\{u,v\}\in E$ for all $u,v\in C$.
The current state of the art for $k$-Clique detection is captured in the following hypothesis, see~\cite{NesetrilP85,AbboudBVW18}.
\begin{hypothesis}[$k$-Clique hypothesis~\cite{AbboudBVW18}]
	\begin{enumerate}
	\item[] 
	\item \emph{Combinatorial version:} For no $k \ge 3$ and $\epsilon > 0$, there is a \emph{combinatorial} algorithm solving $k$-Clique in time $O(n^{k(1-\epsilon)})$.
	\item \emph{Non-combinatorial version:} For no $k \ge 3$ and $\epsilon > 0$, there is an algorithm solving $k$-Clique in time $O(n^{k(\omega/3-\epsilon)})$.
	\end{enumerate}
\end{hypothesis}

One can show the following $k$-Clique-based lower bounds, which in particular give a tight combinatorial lower bound for \emph{dense} NFAs. By \lemref{dense-is-hardest}, this yields a tight combinatorial lower bound for all graph densities.

\begin{proposition} \label{prop:cliqueLB} NFA Acceptance has the following $k$-Clique-based lower bounds:
	\begin{enumerate}
		\item \emph{Combinatorial algorithms:} Unless the combinatorial $k$-Clique hypothesis fails, there are no $\epsilon,\beta > 0$ such that there is an $O( (\ell n^2)^{1-\epsilon})$-time combinatorial algorithm for NFA Acceptance on an $n$-state NFA with $\Theta(n^2)$ transitions and a string of length $\ell=\Theta(n^\beta)$ over $\Sigma=\{0,1,2,3\}$.
		\item \emph{General algorithms:} Unless the $k$-Clique hypothesis fails, there are no $\epsilon,\beta > 0$ such that there is an $O( (\ell n^2)^{\omega/3-\epsilon})$-time algorithm for NFA Acceptance on an $n$-state NFA with $\Theta(n^2)$ transitions and a string of length $\ell=\Theta(n^\beta)$ over $\Sigma=\{0,1,2,3\}$.
	\end{enumerate}  
\end{proposition}

The above lower bound is implicit in the literature.
Specifically, it can be obtained either by adapting the reduction from weighted $k$-Clique in~\cite{BackursT17} to the unweighted case, or as a special case of~\cite[Theorem I.5]{AbboudBBK17}\footnote{In the theorem, simply choose $\alpha_N=\alpha_n=\beta$.}, or as an appropriate generalization of the reduction from Triangle Detection to NFA Acceptance with $\ell \approx n$ in~\cite{PotechinS20}, or as an appropriate generalization of the reduction from Triangle Detection to regular path queries with queries of length $\ell\approx n$ in~\cite{CaselS23}.

For the reader's convenience, we present a simplified proof of this fact.

\begin{proof}[Proof of Proposition~\ref{prop:cliqueLB}]
	We prove the conditional lower bound for combinatorial algorithms. The claim for general algorithms follows from the same reduction.
	
	Assume for the sake of contradiction that there are $\epsilon > 0$ and $\beta >0$ such that NFA Acceptance on a $q$-state NFA with $\Theta(q^2)$ transitions and a string $S$ of length $\ell=\Theta(q^\beta)$ over $\Sigma=\{0,1,2,3\}$ can be solved in time $O((\ell q^2)^{1-\epsilon})=O(q^{(2+\beta)(1-\epsilon)})$.
	
	We reduce from $(2k+k')$-clique where $k> (5+2\beta)/\epsilon$ and $k'=\lfloor \beta k \rfloor$. We let $V=\{0,\dots, n-1\}$, so that we can write $v$ in binary as a $\lceil \log n \rceil$-length bit string, which we call the \emph{node ID} of $v$.
	
	For any $t$, let $\clique(t)$ denote the set of $t$-cliques in $G$.
	We construct the string $S$ as follows:
	\[ S := 2 \circ \bigcirc_{\{v_1,\dots, v_{k'}\}\in \clique(k')} \left( v_1^k\;  v_2^k \; \dots \; v_{k'}^k  \; 3 \; v_1^k\;  v_2^k \; \dots \; v_{k'}^k  \; 2\right),\]
	where $\circ$ denotes concatenation and $v^k$ for $v\in V$ denotes the $k$-fold concatenation of $v$'s node ID.
	
	We turn to constructing the NFA $M$. First, for every $\{u_1,\dots,u_k\}\in \clique(k)$, we construct a clique gadget $CG(u_1,\dots, u_k)$. This gadget will accept a string  of the form $v_1^k\;  v_2^k \; \dots \; v_{k'}^k$ for $\{v_1,\dots, v_{k'}\}\in \clique(k')$ if and only if $\{u_1,\dots, u_k, v_1,\dots, v_{k'}\}$ form a $(k+k')$-clique. We can construct such an NFA with $\tOh(n)$ states and transitions as follows: For any node $u\in V$, it is straightforward to construct an NFA $N(u)$ of size $\tOh(n)$ accepting only the node IDs of vertices $v\in V$ with $\{u,v\}\in E$. We simply introduce a parallel path from starting to accepting state for every neighbor $v$ of $u$. Since every node has at most $n$ neighbors and each ID is of length $O(\log n)$, the size of $N(u)$ is $\tOh(n)$.
	We can now construct $CG(u_1,\dots, u_k)$ by connecting in series the following sequence of NFAs:
	\[ \underbrace{(N(u_1) \; \dots \; N(u_k))\;  \dots \; (N(u_1)\; \dots\; N(u_k))}_{k' \text{ times}}.  \]
	Here, by connecting in series, we mean identifying the single accepting state of each NFA in the sequence with the starting state of the subsequent NFA. Note that the constructed NFA accepts $v_1^k\;  v_2^k \; \dots \; v_{k'}^k$ if and only if $\{u_i,v_j\}\in E$ for every $1\le i \le k$ and $1\le j \le k'$, yielding the desired property. Furthermore, it has $(k+k')\cdot \tOh(n) =\tOh(n)$ states and transitions. 
	
	For each constructed clique gadget $CG(u_1,\dots, u_k)$ with $\{u_1,\dots, u_k\}\in \clique(k)$, we construct a copy $CG'(u_1,\dots, u_k)$. We add transitions from the starting state $s$ of $M$ to the starting state of $CG(u_1,\dots, u_k)$, labeled $2$, and similarly transitions from each $CG'(u_1,\dots, u_k)$ to the accepting state $t$ of $M$, labeled $2$. We also add loop transitions $(q,\sigma,q)$ for each $q\in \{s,t\}$ and $\sigma \in \{0,1,2,3\}$. Finally, for every $\{u_1,\dots, u_k\},\{u'_1,\dots, u'_k\}\in \clique(k)$ such that $\{u_1,\dots, u_k, u'_1,\dots, u'_k\}$ form a $2k$-clique in $G$, we add a transition from the accepting state of $CG(u_1,\dots, u_k)$ to the starting state of $CG'(u'_1,\dots, u'_k)$, labeled 3.
	
	By the structure of $M$ and $S$, the only accepting runs have the following form: $M$ reads a prefix of $S$, branches on start of some substring $y= 2 v_1^k\dots, v_{k'}^k 3 v_1^k\dots, v_{k'}^k 2$ with $\{v_1,\dots, v_{k'}\}\in \clique(k')$ to some $CG(u_1,\dots,u_k)$ and some $CG'(u'_1,\dots,u'_k)$, followed by reading the remaining suffix of $S$ in the accepting state $t$. This run exists if and only if \begin{itemize}
		\item 	$\{v_1,\dots, v_{k'}, u_1,\dots, u_k\}$ forms a $(k+k')$-clique (so that the prefix of $y$ can be traversed with $CG(u_1,\dots, u_k)$),
	\item  $\{u_1,\dots, u_k,u'_1,\dots, u'_k\}$ forms a $2k$-clique (so that there is a transition from $CG(u_1,\dots, u_k)$ to $CG'(u'_1,\dots, u'_k)$), and
	\item $\{v_1,\dots, v_{k'}, u'_1,\dots, u'_k\}$ forms a $(k+k')$-clique (so that the suffix of $y$ can be traversed with $CG'(u'_1,\dots, u'_k)$).
\end{itemize} 
This is equivalent to $\{v_1,\dots, v_{k'},u_1,\dots, u_k, u'_1, \dots, u'_k\}$ forming a $(2k+k')$-clique.

Thus, we have created an NFA Acceptance instance that is equivalent to the given $(2k+k')$-Clique instance. Note that the constructed string $S$ has length $\tOh(n^{k'})$, that $M$ has $O(n^{k+2})$ states, since it consists of $O(n^k)$ clique gadgets of size $\tOh(n)=O(n^2)$, and that the instance can be constructed in linear time in its size.

Note that by our choice of $k'=\lfloor \beta k\rfloor \le \beta k$, we have $|S| \le \tOh(n^{k'}) \le O((n^{k+2})^\beta)$. Thus, by appending additional 2's to $S$ and isolated states to $M$, we can ensure that $|S| = \Theta(q^\beta)$ where $q=\Theta(n^{k+2})$ is the number of states of $M$. By adding transitions between the newly added isolated states we can also ensure that the number of transitions is $\Theta(q^2)$.

Thus, our $O(q^{(2+\beta)(1-\epsilon)})$-time combinatorial algorithm for NFA Acceptance on $q$-state NFAs with $\Theta(q^2)$ transitions and strings of length $\Theta(q^\beta)$ solves any given $k$-Clique instance in time $O(n^{(k+2)(2+\beta)(1-\epsilon)})$. Note that 
\begin{align*}
 (k+2)(2+\beta)(1-\epsilon) & = (2k +\beta k) + (4+2\beta) - (k+2)(2+\beta)\epsilon \\
 & \le (2k +\beta k) + (4+2\beta) - k\epsilon \\
 & \le (2k+k') + (5+2\beta) - k\epsilon \\
 & < 2k + k', 
 \end{align*}
 since $k > (5+2\beta)/\epsilon$. Thus, there exists some $\epsilon' > 0$ such that we can solve $(2k+k')$-Clique in time $O(n^{2k+k'-\epsilon'})$ by a combinatorial algorithm, contradicting the combinatorial $k$-Clique hypothesis.
\end{proof}

\subsection{Co-nondeterministic Algorithm}

We say that a problem $P_A$ has a $t(n)$-time verifier if $P_A$ can be solved in nondeterministic time $t(n)$ and also its complement problem $\overline{P_A}$ can be solved in nondeterministic time $t(n)$, i.e., $P_A$ is in $\mathrm{NTIME}[t(n)] \cap \mathrm{coNTIME}[t(n)]$. Under NSETH~\cite{CarmosinoGIMPS16}, there exists no deterministic fine-grained reduction from Satisfiability to $P_A$ that would establish a SETH-based lower bound of $t(n)^{1+\delta}$ for any $\delta > 0$.

For the reader's convenience, we give a simple verifier for Colored Walk, and thus NFA Acceptance. This verifier is faster than the one for the generalization to the Viterbi Path problem designed in the arXiv version of~\cite{BackursI16}, and simplifies and extends the verifier one could obtain as a consequence of our reduction in Section~\ref{sec:cfl-reach} combined with~\cite{ChistikovMS22}. 

\begin{proposition}\label{prop:nondet}
	Directed 2-Edge Colored Walk has an $O(n^\omega + \ell n^{\omega-1})$-time verifier.
\end{proposition}

As a consequence, we obtain verifiers with the same running time (up to subpolynomial factors) for all equivalent Colored Walk formulations and the NFA Acceptance problem over $\Sigma$ of size $n^{o(1)}$.

Thus, assuming NSETH, there cannot be any deterministic reduction from Satisfiability (or Orthogonal Vectors) that establishes a tight conditional lower bound for NFA Acceptance with $\alpha > \omega-1$ and $\beta> \omega-\alpha$; this includes in particular the important setting of $\alpha=2$ and $\beta = 1$. This gives a justification why until now, no tight SETH-based lower bound for dense NFAs could be established.

\begin{proof}[Proof of Proposition~\ref{prop:nondet}]
	Consider a Directed 2-Edge Colored Walk instance $G=(V,E)$ with edge colors $c:E\to \{1,2\}$, distinguished nodes $s,t\in V$ and color sequence $c_1,\dots, c_\ell$. Define $A_c$ as the $n\times n$ matrix with
	\[A_c[u,v] = \begin{cases} 1 & \text{if } (u,v)\in E \text{ and } c(u,v) = c\\
	0 & \text{otherwise.} \end{cases} \]
	Furthermore, define $x_0 \in \{0,1\}^n$ as the indicator vector representing the starting state $s$, i.e., $x_0[v] = 1$ if and only if $v=s$. For all $1\le i \le \ell$, we define
	\begin{equation} \label{eq:tocheck}
	x_i = A_{c_i}^T \cdot x_{i-1}, 
	\end{equation}
	where $A^T$ denotes the transpose of $A$ and $A^T \cdot x$ denotes the Boolean matrix-vector product of $A^T$ and $x$. Observe that $x_i$ is the indicator vector with $x_i[v] = 1$ if and only if there exists a walk with color sequence $c_1,\dots, c_i$ from $s$ to $v$.
	
	To obtain nondeterministic algorithms for the problem and its complement problem, we simply need to guess $x_1,\dots, x_\ell$, verify that they have been guessed correctly, and check whether $x_\ell[t] = 1$, or $x_\ell[t] = 0$, respectively.
	
	To verify that the $x_i$'s have been guessed correctly, we batch the equalities~\eqref{eq:tocheck} that we need to check into two Boolean matrix products: Specifically, for each $c\in \{1,2\}$, let $X_c$ denote the matrix containing all $x_{i}$'s with $c_i = c$ as columns, and let $X'_c$ denote the matrix containing all $x_{i-1}$'s with $c_i = c$ as columns, in the same order. For each $c\in \{1,2\}$, we need to check that
	\[ X_c = A^T X'_c. \]
	Note that this can be done via two (integer) matrix multiplications of two matrices with dimensions $n\times n$ and $n\times \ell$. If $\ell \le n$, this can be done with a single square matrix multiplication, otherwise, we can do this with $\lceil \frac{\ell}{n}\rceil$ square matrix multiplications. This results in a total verification time of $O((1+\frac{\ell}{n})n^\omega)=O(n^\omega + \ell n^{\omega-1})$, as desired.
\end{proof}

\section{Failed Algorithmic Approaches for NFA Acceptance}
\label{sec:natural-approaches}

In this section we consider some natural algorithmic approaches to NFA Acceptance and discuss why they do not falsify the NFA Acceptance hypothesis. Instead of directly working with the NFA Acceptance problem, we consider Directed Edge-2-Colored Walk (or short: Colored Walk). This is without loss of generality, as shown by the equivalence established in \lemref{equivalencehypo}.

\paragraph{Algorithms Using Fast Matrix Multiplication?}
Note that in the special case of only one color, i.e, $c_1= \ldots = c_\ell$, Colored Walk can be solved by computing the $\ell$-th matrix power $A^{\ell}$ of the adjacency matrix $A$, and then checking whether the $(s,t)$-entry is non-zero. This takes time $O(n^\omega \log \ell)$, which is much faster than $O(m \ell)$ if $\ell$ is large. 

If we could achieve the same running time $O(n^\omega \log \ell)$ in the case of two colors, then we would falsify the NFA Acceptance hypothesis (for any setting $\alpha \in [1,2]$ and $\beta > \omega - \alpha$, by \lemref{equivalencehypo}). However, as we argue next, the natural generalization of the above algorithm to two colors fails. 
Consider adjancency matrices $A^{(1)},A^{(2)}$, where the $(u,v)$-entry of $A^{(c)}$ is 1 if there is an edge of color~$c$ from $u$ to $v$, and 0 otherwise. Then we would want to compute the matrix $A^{(c_1)} A^{(c_2)} \cdots A^{(c_\ell)}$ and check whether its $(s,t)$-entry is non-zero, to solve Directed Edge-2-Colored Walk. However, in the worst case no two subsequences of $c_1,\ldots,c_\ell$ of length $\log(\ell)$ are equal, and thus we cannot save more than log-factors by precomputing products of the form $A^{(c_i)} \cdots A^{(c_j)}$ -- in contrast to the special case of only one color. Hence, we are essentially forced to compute the product $A^{(c_1)} A^{(c_2)} \cdots A^{(c_\ell)}$ one-by-one over $\ell$ steps, resulting in time $O(n^\omega \ell)$. This is worse than the simple $O(m \ell)$-time algorithm. In other words, the natural generalization of the matrix-multiplication-based algorithm for one color fails for two colors.

\paragraph{Parallel Algorithms for Colored Walk}
The matrix multiplication approach sketched above is nevertheless useful to design \emph{parallel} algorithms for Colored Walk. Indeed, matrix multiplication can be easily parallelized, e.g., on the PRAM one can multiply two $n \times n$ matrices with depth $O(\log n)$ and work $O(n^3)$. To compute the product $A^{(c_1)} A^{(c_2)} \cdots A^{(c_\ell)}$ we can recursively compute the product $A^{(c_1)} \cdots A^{(c_{\lfloor \ell/2 \rfloor})}$ and recursively compute the product $A^{(c_{\lfloor \ell/2 \rfloor+1})} \cdots A^{(c_{\ell})}$ and then multiply the results. If we perform both recursive calls in parallel, then we obtain a PRAM algorithm with depth $O(\log (n) \log (\ell))$ and work $O(n^3\ell)$. From the final matrix product $A^{(c_1)} \cdots A^{(c_\ell)}$ we can then read off the answer as the $(s,t)$-entry. 
This shows that Colored Walk (and in a similar way also NFA Acceptance) has an efficient parallel algorithm. Note that the total work $O(n^3\ell)$ of this PRAM algorithm is much more than $O(m \ell)$, and hence this parallel algorithm does not violate the NFA Acceptance hypothesis.

\paragraph{Inherently Sequential?}
The above parallel algorithm has much higher work than the sequential time complexity $O(m \ell)$. It thus seems plausible that any parallel algorithm for Colored Walk with work $\tOh(m \ell)$ has depth at least $\ell^{1-o(1)}$. In other words, it is plausible that \emph{work-optimal algorithms for Colored Walk are inherently sequential}. 

This intuition is consistent with our current knowledge of the problem. 
However, there are no tools available to prove this intuition, even conditionally.
The only tool that we have to show that certain parallel algorithms are unlikely for a problem $X$ is to show that $X$ is P-hard. However, since we have seen a PRAM algorithm with polylogarithmic depth for Colored Walk, it is very unlikely that Colored Walk is P-hard, and thus this tool is not applicable for this problem. Hence, currently there are no tools to provide evidence for the intuition that work-optimal algorithms for Colored Walk are inherently sequential.

\section{Conclusion}
\label{sec:conclusion}

In this work, we have posed the NFA Acceptance hypothesis, discussed implications of a proof or refutation and considered its connections to standard assumptions in fine-grained complexity theory. 
At the very least, the NFA Acceptance hypothesis should be understood as a technical challenge towards progress for many interesting problems. Specifically, without refuting it, we cannot expect:
\begin{itemize}
	\item Context-Free Language Reachability in time $O(n^{3-\epsilon})$,
	\item the Word Break problem in time $O(m+nm^{1/3-\epsilon})$,
	\item approximating the Viterbi path in time $O((m\ell)^{1-\epsilon})$,
	\item pattern matching in string-labeled graphs in time $O((m\ell)^{1-\epsilon})$,
	\item regular path queries in graph databases in time $O((|q|\cdot m)^{1-\epsilon})$, and
	\item polynomial improvements for any data structure problems and dynamic problems with tight OMv-hardness, including problems from graph algorithms~\cite{HenzingerKNS15,Dahlgaard16,AbboudD16,HenzingerW21}, string algorithms~\cite{CliffordGLS18,KempaK22,CliffordGK0U22}, computational geometry~\cite{DallantI22,LauR21}, linear algebra~\cite{JiangPW23}, formal languages~\cite{GibneyT21}, and database theory~\cite{BerkholzKS17,CaselS23,KaraNOZ23}.
\end{itemize}
We leave the following open problems:
\begin{itemize}
	\item Refute the NFA Acceptance hypothesis or prove that it is implied by any standard fine-grained hypothesis.
	\item Prove an equivalence of NFA Acceptance and Directed $n^\gamma$-Edge-Colored Walk for any $\gamma > 0$, ideally for all $0 < \gamma \le 2$.
	\item Prove an equivalence of NFA Acceptance with and without $\eps$-transitions for densities $m=\Theta(n^{\alpha})$ with $1\le \alpha < 2$.
\end{itemize}

\paragraph{Acknowledgements}
We thank Virginia Vassilevska Williams, Thatchaphol Saranurak, and Rupak Majumdar for helpful discussions on the NFA Acceptance hypothesis.
We also thank Thatchaphol Saranurak for popularizing this hypothesis by tweeting the open problem on Colored Walk and its implication for OMv, crediting one of the authors of this paper.

\printbibliography

\end{document}